\documentclass[journal,onecolumn,draftclsnofoot]{IEEEtran}
\usepackage{placeins}
\usepackage{bbm}
\ifCLASSINFOpdf
  \usepackage[pdftex]{graphicx}
  \usepackage{amsfonts}
  \usepackage{xcolor}

  \usepackage{bbm}
  \usepackage{graphicx}
  \usepackage{makecell}
  \usepackage{hhline}
  \usepackage{multirow}
  \graphicspath{{figures/}}

  \usepackage{booktabs} 
  \usepackage{array}
\usepackage{subcaption}

  \usepackage{amsmath,amssymb,amsthm}
  \usepackage{cite}
  \usepackage[colorlinks=true, linkcolor=blue, citecolor=blue, urlcolor=blue]{hyperref}

  \newtheorem{proposition}{Proposition}[section]

\newtheorem{corollary}{Corollary}[section]

\usepackage{xcolor}

\usepackage{titlesec}
\titlespacing{\section}
{0pt}
{1.3ex plus 0.2ex minus 0.2ex}
{0.5ex plus 0.1ex}

\titlespacing{\subsection}
{0pt}
{0.9ex plus 0.2ex minus 0.1ex}
{0.3ex plus 0.1ex}

\titleformat{\subsubsection}[runin]
{\normalfont\normalsize\itshape}
{\arabic{subsubsection})}
{0.5em}
{}
[:]

\titlespacing{\subsubsection}
{0pt}
{0.7ex plus 0.1ex minus 0.1ex}
{0.5em}

\else
\fi

\begin{document}
\bstctlcite{IEEEexample:BSTcontrol}
\title{Post-Hoc Conformal Prediction \\for Reliable Wireless Communications}
\author{Xin~Su,
        Meiyi~Zhu,~\IEEEmembership{Member,~IEEE,}
        Osvaldo~Simeone,~\IEEEmembership{Fellow,~IEEE,}
        and~Carlo~Fischione,~\IEEEmembership{Fellow,~IEEE}
\thanks{Xin Su and Carlo Fischione are with the School of Electrical Engineering and Computer Science, KTH Royal Institute of Technology, 10044 Stockholm, Sweden (e-mail: xisu@kth.se; carlofi@kth.se).
Meiyi Zhu is with the Department of Engineering, King's College London, London WC2R 2LS, U.K. (e-mail: meiyi.1.zhu@kcl.ac.uk).
Osvaldo Simeone is with the Institute for Intelligent Networked Systems, Northeastern University London, London E1 8PH, U.K. (e-mail: o.simeone@northeastern.edu).
The work of X. Su and C. Fischione was sponsored by the KTH DF research center and by the SSF SAICOM project. The work of M. Zhu and O. Simeone was supported by an Open Fellowship of the EPSRC (EP/W024101/1). The work of O. Simeone was also supported by EPSRC (EP/X011852/1) and ERC (No. 101198347).
}
}



\maketitle

\begin{abstract}
Deploying artificial intelligence (AI) in high-stakes wireless applications such as autonomous transportation requires guarantees of reliable operation. Such guarantees can often be obtained by designing systems that act on a set of predictions via conservative policies catering to all possible outcomes within the set. For instance, in location-based beam selection, a base station may identify a set of plausible locations and select a beam that ensures high capacity over the set. The miscoverage probability of the set with respect to the true outcome (e.g., the true user location) then  quantifies the outage probability, while the size of the set determines the final performance or resource budget (e.g., the transmission capacity). Conventional conformal prediction (CP) applies when the target miscoverage level is prescribed in advance. However, in practical wireless systems, prediction sets may instead be selected under prescribed operational constraints, requiring the resulting miscoverage probability to be quantified post hoc. This paper develops a formal statistical framework for the data-driven selection of prediction sets that provides a reliable estimate of the resulting miscoverage probability. The methodology builds on backward CP and probably approximately correct CP, yielding distribution-free reliability guarantees on the miscoverage probability. We apply the framework on three distinct wireless applications: narrowband interference detection, near-field localization, and codebook-based beam identification. Numerical results validate the reliability guarantees and show that the
proposed post-hoc conformal methods achieve accuracy comparable to the
na\"{\i}ve probability-based approach across all considered applications.
\end{abstract}


\IEEEpeerreviewmaketitle

\section{Introduction}
\label{sec:intro}
\subsection{Motivation}

Artificial intelligence (AI) models are deeply embedded in decision processes, with their predictions informing consequential downstream operations. This is particularly the case in high-stakes wireless applications such as industrial automation, autonomous transportation, and remote health monitoring, in which a connectivity failure may have direct physical or safety consequences. For instance, in autonomous transportation, AI-based predictions of channel and network conditions can support radio-resource allocation and connectivity management in highly dynamic vehicular systems~\cite{ye2018machine}. Erroneous predictions may then delay or disrupt the delivery of safety-critical information, potentially compromising vehicle control and endangering autonomous operation~\cite{garcia2021tutorial}. In such high-stakes settings, a fundamental concern extends beyond predictive accuracy to how reliably the predictions support the resulting operations.

Achieving greater reliability, however, generally requires safeguarding against a broader range of possible outcomes, which may incur operational costs such as increased consumption of communication, computation, or hardware resources, or reduced system performance~\cite{popovski2018wireless}.
At the extreme, one could maximize reliability by accommodating every possible output, which would render the prediction uninformative while incurring prohibitive operational costs.
Practical systems, however, cannot pursue reliability without regard to such costs. Instead, resource budgets or performance requirements are often prescribed in advance, making it necessary to reliably quantify the predictive uncertainty that remains under the imposed operational constraints~\cite{sze2017efficient}.

A natural way of assessing this uncertainty level is through the confidence reported by the underlying AI model, which reflects its own assessment of predictive uncertainty.
The reliability of this self-assessment depends on its \textit{calibration}, i.e., the agreement between the confidence reported by the model and the true probability of the corresponding predictive outcome~\cite{cohen2023calibrating,guo2017calibration,simeone2026decision}.
Calibration guarantees, however, may rely on strong assumptions about the ground-truth, unknown data-generation mechanism~\cite{masegosa2020learning,zecchin2023robust}, while practical models can remain poorly calibrated, with overconfidence leading to underestimated predictive uncertainty~\cite{romano2020classification,huang2025calibrating}.
 This raises the question of how predictive uncertainty can be reliably quantified under prescribed operational constraints when the model confidence itself may be inaccurate.

\subsection{Reliable Uncertainty Quantification for Wireless Systems}

AI-based predictions, together with their associated predictive uncertainty, are increasingly leveraged to support adaptive wireless operations under dynamic and uncertain conditions.
A common way to translate such predictions into robust operation design is through a \textit{set predictor}, which specifies a set of plausible outputs, such as user locations, channel realizations, or network conditions, on which the downstream operation then acts via a conservative policy catering to all outcomes contained in the set~\cite{cohen2023calibrating,kiyani2024length,chenreddy2022data}.


\begin{figure}
    \centering
    \includegraphics[width=0.8\linewidth]{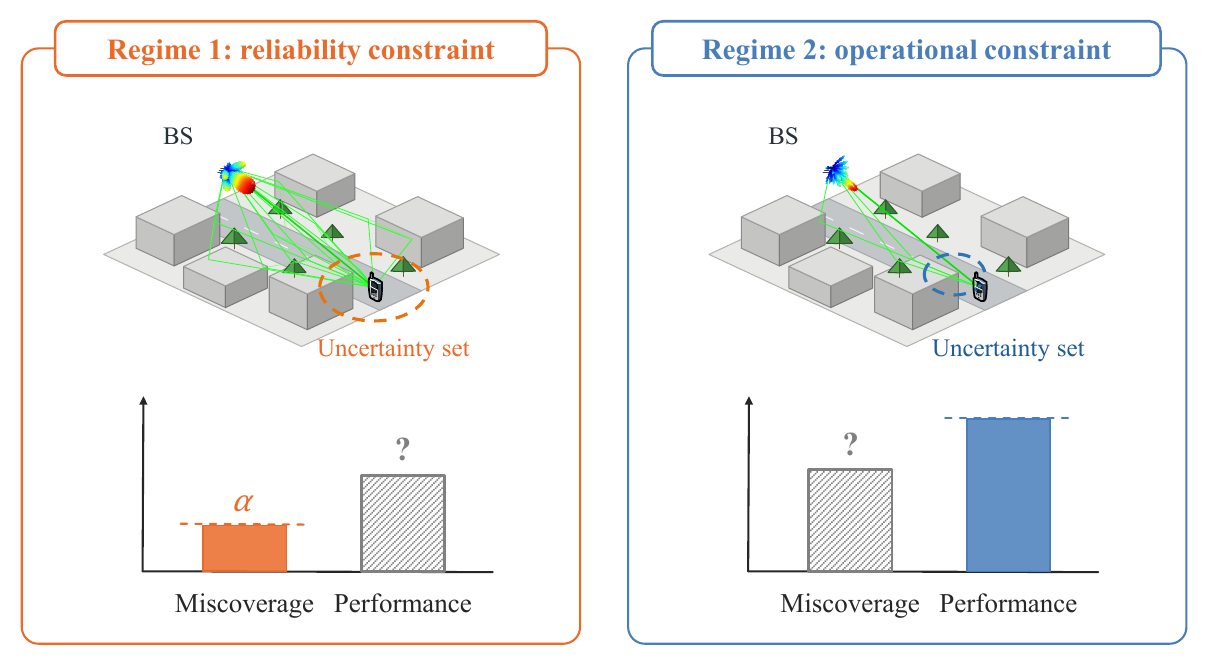}
\caption{Uncertainty-aware robust design in wireless systems can adopt two alternative criteria, illustrated here for location-based beam selection. 
The first desigrescribes a target miscoverage level $\alpha$ for the uncertainty set on the user’s location, and optimizes the beam to cater to all locations within the set. In the second regime, instead one specifies desired properties of the uncertainty set, e.g., its size and associated beamforming gain, and obtains an estimate $\hat \alpha$ of the resulting miscoverage probability. While the first regime can be formally addressed with CP~\cite{vovk2005algorithmic,su2025conformal}, the second regime requires alternative methodologies that are studied in this paper.}
    \label{fig:intro_example}
\end{figure}

As a running example, consider location-based beam selection. Rather than committing to a single estimate of the user position, the base station identifies a set of plausible locations and selects a beam that guarantees high capacity for every location within the set \cite{su2025conformal}. Two quantities then summarize the resulting operation. First, the \textit{miscoverage probability} of the set with respect to the true outcome, namely the probability that the true user location falls outside the prediction set, quantifies the outage probability of the downstream operation~\cite{kiyani2025decision}. Second, the \textit{size of the set} determines the attainable performance, or the required resource budget: a larger location set must be served by a wider beam and hence yields a lower transmission capacity, whereas a smaller set supports a higher capacity at the price of an increased risk of outage \cite{su2025conformal}.

As illustrated in Fig.~\ref{fig:intro_example}, this tension admits two complementary design regimes. In the first, the target reliability level is prescribed, and the set predictor is designed so as to limit its miscoverage probability. For example, in robust beamforming, an outage probability of at most $\alpha$ can be ensured by optimizing the beamformer over a channel uncertainty set with miscoverage probability at most $\alpha$~\cite{su2025conformal}. In the second regime, which is the focus of this work, operational constraints, such as transmit-power budgets, latency requirements, hardware limitations, or a minimum transmission capacity, instead govern the design of the set predictor by restricting the number or range of outputs included in the set~\cite{lu2020positioning}. In such settings, the prediction set is selected in a data-driven manner so as to meet the constraint, and it cannot generally guarantee any prescribed reliability level. Accurate quantification of its miscoverage probability, and hence of the outage probability of the downstream operation, thus becomes necessary for assessing the risk incurred by the system.



\subsection{Conformal Prediction and Backward Prediction}
A principled way to obtain reliable uncertainty guarantees without relying on well-calibrated predictive probabilities is conformal prediction (CP), which calibrates prediction sets to achieve prescribed coverage. For a prescribed miscoverage level $\alpha$, conventional CP uses calibration data to satisfy  \textit{finite-sample} marginal coverage guarantees of at least $1-\alpha$, irrespective of the underlying data distribution~\cite{vovk2005algorithmic,angelopoulos2021gentle}. Therefore, CP provides an ideal solution for the first design regime mentioned above, which is characterized by fixed reliability levels. Based on this property, CP has also been recently applied to wireless systems, including to symbol demodulation~\cite{cohen2023calibratingc},  beamforming~\cite{su2025conformal, zhang2025calibrating}, network control~\cite{hou2025what}, and distributed inference~\cite{zhu2024federated,zhu2025conformal}. 

However, CP is not directly applicable to the design regime of interest in this work, in which the prediction set must be selected to meet an operational requirement and the associated reliability  level is unknown a priori. Recent work has extended CP to post-hoc settings, reversing the conventional workflow by allowing the miscoverage level to be determined after calibration rather than fixed a priori while retaining formal reliability guarantees.
\textit{Backward conformal prediction} (BCP) realizes this flexibility, allowing the prediction set to satisfy a prescribed, possibly input-dependent operational constraint and inferring the corresponding miscoverage level~\cite{gauthier2025values,gauthier2025backward}. 
Recent advances have further tightened BCP through data-dependent score transformations~\cite{liu2026improving} and extended post-hoc conformal inference to fuzzy prediction sets~\cite{koning2025fuzzy}. Related post-hoc settings have also been addressed by \textit{probably approximately correct (PAC) CP}~\cite{sarkar2023post} and inverse conformal risk control (CRC)~\cite{zhou2025calibrating}. Table~\ref{tab:posthoc_comparison} summarizes the above post-hoc conformal approaches.
Despite these recent developments, to the best of our knowledge, post-hoc conformal frameworks remain largely unexplored in wireless systems.

\begin{table*}[t]
\centering
\caption{Comparison of related post-hoc conformal approaches.}
\label{tab:posthoc_comparison}
\renewcommand{\arraystretch}{1.15}
\begin{tabular}{lcccccc}
\toprule
\textbf{Method}
& \shortstack{\textbf{Prescribed}\\\textbf{size constraint}}
& \shortstack{\textbf{Input-dependent}\\\textbf{estimate}}
& \shortstack{\textbf{Calibration-dependent}\\\textbf{estimate}}
& \shortstack{\textbf{Average-ratio}\\\textbf{guarantee}}
& \shortstack{\textbf{PAC}\\\textbf{guarantee}}
& \shortstack{\textbf{Wireless}\\\textbf{applications}} \\
\midrule

BCP~\cite{gauthier2025values,gauthier2025backward}
& \checkmark
& \checkmark
&
& \checkmark
&
& \\

ST-BCP~\cite{liu2026improving}
& \checkmark
& \checkmark
&
& \checkmark
&
& \\

Fuzzy CP~\cite{koning2025fuzzy}
&
& \checkmark
&
& \checkmark
&
& \\

$\delta$-PAC CP~\cite{sarkar2023post}
&
&
& \checkmark
&
& \checkmark
& \\

Inverse CRC~\cite{zhou2025calibrating}
&
&
& \checkmark
&
& \checkmark
& \\

\midrule
\textbf{This paper}
& \checkmark
& \checkmark
& \checkmark
& \checkmark
& \checkmark
& \checkmark \\

\bottomrule
\end{tabular}
\end{table*}

\subsection{Main Contributions}

Our conference study~\cite{su2026reliable} applied BCP to reliable miscoverage estimation for budget-constrained narrowband interference detection under an input-dependent operational constraint. Building on this initial application, the present work provides the first systematic investigation of post-hoc conformal frameworks for the data-driven selection of prediction sets, and for the reliable quantification of the resulting miscoverage probability, under operational constraints in wireless systems. Our main contributions are summarized as follows.

\begin{itemize}

\item We formulate the problem of selecting a prediction set subject to application-specific operational constraints, represented by a generic set functional that captures requirements on resource cost and downstream performance, while simultaneously providing a reliable estimate of the miscoverage probability of the selected set and hence of the outage probability of the downstream operation. The formulation accommodates both input-dependent and average constraints, for which reliable miscoverage estimation is posed under the average-ratio and probably approximately correct (PAC) criteria, respectively.

\item We develop post-hoc conformal procedures for reliable miscoverage estimation by leveraging BCP~\cite{gauthier2025backward}, score-transformation backward conformal prediction (ST-BCP)~\cite{liu2026improving}, and $\delta$-PAC CP~\cite{sarkar2023post}, and establish their distribution-free reliability guarantees. BCP and ST-BCP yield input-specific estimators for input-dependent constraints under the average-ratio criterion, whereas $\delta$-PAC CP yields a calibration-dependent estimator for average constraints under the PAC criterion.

\item Finally, we instantiate the proposed framework in three wireless applications with distinct operational requirements, namely narrowband interference detection, near-field localization, and codebook-based beam identification, in which the size of the prediction set maps to a resource budget or to an attainable performance level. Numerical results validate the corresponding reliability guarantees and show that the conformal estimators achieve estimation accuracy comparable to that of the na\"{\i}ve probability-based approach.

\end{itemize}

The conference version~\cite{su2026reliable} of this work considered only BCP for narrowband interference detection. In contrast, this work develops a general framework covering both input-dependent and average operational constraints, and systematically investigates multiple post-hoc conformal methods, namely BCP, ST-BCP, and $\delta$-PAC CP. Furthermore, it extends the study to near-field localization and codebook-based beam identification, with additional numerical evaluations.

\subsection{Further Related Work}

Beyond the conformal approaches discussed above, a complementary line of work seeks to improve the reliability of the predictive probabilities reported by AI models. 
We briefly review two representative directions, namely Bayesian and ensemble learning, and post-processing calibration. 
These approaches aim to improve the calibration of model-based uncertainty estimates, but generally do not provide distribution-free finite-sample reliability guarantees.

\subsubsection{Bayesian and Ensemble Learning}
Bayesian and ensemble learning account for model uncertainty by incorporating uncertainty into the model learning and prediction process.
Specifically, Bayesian learning places a prior over the model parameters and infers a posterior from the training data, with predictions obtained by averaging over the posterior~\cite{simeone2022machine}. 
Ensemble learning instead trains multiple models and aggregates their predictions, with disagreement among ensemble members used to characterize model uncertainty~\cite{dong2020survey}. 
In wireless systems, Bayesian and ensemble approaches have been investigated for tasks including localization~\cite{tedeschini2024real}, resource allocation~\cite{zhang2026efficient}, and digital-twin maintenance~\cite{jankov2026reliability}.
Despite explicitly accounting for model uncertainty, these approaches remain dependent on the assumed model family and learning procedure, and may therefore still yield poorly calibrated predictive distributions under model misspecification~\cite{sabanovic2026calibration}.

\subsubsection{Post-Processing Calibration}

A lightweight alternative is to use separate calibration data to recalibrate the predictive distribution of a pretrained model without retraining it.
Representative approaches include Platt scaling and temperature scaling, which apply simple parametric transformations to model outputs~\cite{guo2017calibration,platt1999probabilistic}, and isotonic regression, which learns a nonparametric monotone mapping~\cite{zadrozny2002transforming}.
More flexible methods include Dirichlet calibration for multiclass predictions~\cite{kull2019beyond} and parameterized temperature scaling with prediction-dependent parameters~\cite{tomani2022parameterized}.
In wireless systems, post-processing calibration has been studied for automatic modulation classification~\cite{judah2024demonstrating} and resource allocation~\cite{raina2025trust}.
These approaches, however, do not provide distribution-free calibration guarantees and may be sensitive to limited calibration data and the choice of calibration mapping~\cite{cohen2023calibrating}.

\subsection{Organization}
The remainder of this paper is organized as follows. Sec.~\ref{sec:prob_def} introduces the background and formulates the problem of prediction-set selection and reliable miscoverage estimation under operational constraints. Sec.~\ref{sec:bcp_proc} reviews the na\"{\i}ve miscoverage estimation method and presents the post-hoc conformal procedures based on BCP, ST-BCP, and $\delta$-PAC CP. Secs.~\ref{sec:intf_det}--\ref{sec:bf} investigate their applications to narrowband interference detection, near-field localization, and codebook-based beam identification, respectively. Finally, Sec.~\ref{sec:conclusion} concludes the paper.

\section{Background and Problem Definition}
\label{sec:prob_def}
This section formulates the problem of designing prediction sets subject
to application-specific operational constraints and  reliable miscoverage
estimates.

\subsection{Background}
We consider a standard inferential setting in which a generic test pair $(x,y) \sim p(x,y)$  consists of an observed feature $x \in \mathcal{X}$, such as a baseband signal collected by the radio access network, and of the associated target value $y \in \mathcal{Y}$, e.g., the location of the transmitter or the presence/absence of interference in a given subband (see Fig.~\ref{fig:set_id}).
A pre-trained predictor is available that outputs an estimate $\hat{p}(\cdot| x)$ of the conditional distribution $p(y|x)=p(x,y)/p(x)$ of the target variable $y$.

In a conventional point-estimation approach, the predictor
applies the maximum a posteriori (MAP), returning the most likely
output
\begin{equation}
\label{eq:point_est}
    \hat y_x^{\mathrm{MAP}}
    =
    \arg\max_{y\in\mathcal Y}\hat p(y|x).
\end{equation}
Beyond point decisions, downstream operations may benefit from retaining
multiple plausible outputs.
We therefore focus on applications in which the desired output is a set $\mathcal{C}(x)$ of plausible values for variable $y$ given input $x$. 
For example, for beam identification in millimeter wave or terahertz links, the base station (BS) wishes to determine a limited number of candidate beam directions to probe via training and feedback from the user equipments (UEs)~\cite{othman2026diffusion}. Furthermore, as shown in \cite{kiyani2025decision}, any risk-averse decision maker can without loss of
optimality first evaluate a prediction set $\mathcal{C}(x)$, and then act by optimizing the worst-case outcome in set $\mathcal{C}(x)$ (see also~\cite{simeone2026decision}). 

Given a probabilistic predictor
$\hat{p}(y|x)$, in order to extract a prediction set $\mathcal{C}(x)$, we introduce a negatively-oriented score function
$s:\mathcal X\times\mathcal Y\to\mathbb R$ to quantify how well a
candidate output $y\in\mathcal Y$ matches the input $x$, where a smaller score
indicates a better match. 
For example, one can set the score as the complement of the predictive probability $s(x,y)=1-\hat{p}(y|x)$, or more generally, as any decreasing function $\phi(\cdot)$ of the predicted probability $\hat{p}(y|x)$. Examples include the inverse power (IP) score with hyperparameter $\beta>0$ ~\cite{gauthier2025values,kiyani2024length} 
\begin{equation}
\label{eq:nc_score_power}
    s(x,y)=\frac{1}{\hat{p}(y|x)^{\beta}},
\end{equation}
or the log-loss (LL) score~\cite{gauthier2025backward}
\begin{equation}
\label{eq:nc_score_log}
    s(x,y)= -\log(\hat{p}(y|x)).
\end{equation}

\begin{figure*}[t]
    \centering
\includegraphics[width=\linewidth]{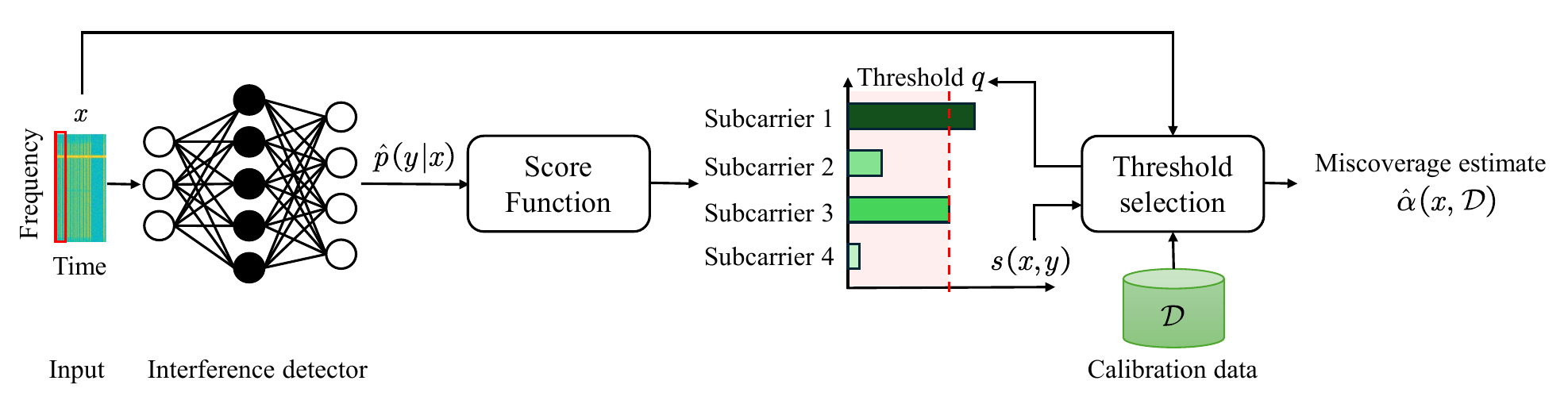}
    \caption{Prediction sets with controllable data-adaptive size and miscoverage estimates illustrated for the task of narrowband interference (NBI) detection. A pre-trained interference detector processes the input data $x$ and outputs predictive probabilities $\hat{p}(y|x)$, which are mapped to scores $s(x,y)$. Based on the calibration dataset $\mathcal{D}$ and the input $x$, a threshold $q$ is selected to construct a prediction set $\mathcal{C}_{q}(x)$ via \eqref{eq:score_threshold_set} while controlling the size and providing a reliable estimate of the miscoverage level $\hat{\alpha}(x,\mathcal{D})$ of the set $\mathcal{C}_{q}(x)$.}
    \label{fig:set_id}
\end{figure*}

Given any such score, a natural way of identifying a prediction set is to
include all candidate outputs whose scores are below a prescribed threshold
$q$, i.e.,
\begin{equation}
    \mathcal C_q(x)
    =
    \left\{
    y\in\mathcal Y:
    s(x,y)< q
    \right\}.
    \label{eq:score_threshold_set}
\end{equation}
Assume access to a calibration dataset $\mathcal{D}=\{(x_i,y_i)\}_{i=1}^{N_{\mathrm{cal}}}$, with calibration samples $(x_i,y_i)$ independently drawn from the data distribution $p(x,y)$. For a fixed target miscoverage probability $\alpha$, conventional CP \cite{vovk2005algorithmic} uses this dataset to select the threshold $q$ such that the set in \eqref{eq:score_threshold_set} contains the true label $y$ with probability no smaller than $1-\alpha$, i.e.,
\begin{equation}
    \mathrm{Pr}\left(y\in \mathcal{C}_{q}(x)\right)\ge 1-\alpha.
\label{eq:cp_marginal_cvg}
\end{equation}
The probability $\mathrm{Pr}(\cdot)$ in \eqref{eq:cp_marginal_cvg} is taken over both test pair $(x,y)\sim p(x,y)$ and calibration data $\mathcal{D}\overset{\mathrm{i.i.d.}}{\sim}p(x_i,y_i)$, which are assumed
to be mutually independent. 
A fundamental limitation of conventional CP in the context of the applications of interest here is the lack of control over the operational properties of the prediction set, such as its size, resource cost, or downstream performance. In fact, the marginal coverage condition~\eqref{eq:cp_marginal_cvg} provides no information about its resource cost or downstream utility.


To capture such operational properties, we introduce a non-negative set functional $g(\cdot):2^{\mathcal Y}\to\mathbb R_{+}$ that maps a prediction set to a resource or utility-related loss. A larger set $\mathcal C_q(x)$ entails a larger loss, as it reflects a situation of higher uncertainty. Accordingly, the function $g(\cdot)$ is non-decreasing in the cardinality of the set $\mathcal C_q(x)$. For example, in settings such as downstream beam identification (which we will discuss thoroughly later), the loss $g(\cdot)$ may be expressed as
\begin{equation}
    g\left(\mathcal C\right) = \sum_{y\in \mathcal C} w_y,
    \label{eq:weighted_cardinality}
\end{equation}
for some set of label-specific weights
$w_y\ge 0$, where $w_y$ is the cost of label $y$. The loss~\eqref{eq:weighted_cardinality} recovers the set cardinality $g\left(\mathcal C\right)=|\mathcal C|$ for weights $w_y=1$ for all $y \in \mathcal{Y}$. Alternatively, for decision makers with risk-averse objectives, $g(\cdot)$ can be expressed as the worst-case loss, captured by weight $w_y$ within the set $\mathcal C$~\cite{kiyani2025decision} 
\begin{equation}
    g\left(\mathcal C\right) = \max_{y\in\mathcal C} w_y,
    \label{eq:worst_case_utility}
\end{equation}

Given the functional $g(\cdot)$, we seek the largest prediction set $\mathcal C_q(x)$ subject to a prescribed constraint on the value of the loss $g(\mathcal C_q(x))$ by selecting the threshold $q$ in~\eqref{eq:score_threshold_set} adaptively based on the test input $x$ and the calibration data $\mathcal D$. We consider two main options:
\begin{enumerate}
    \item \emph{Input-dependent operational constraint}: We fix a priori a maximum constraint level $G_{\max}$, which is enforced for every input $x$ by selecting the threshold as
    \begin{equation}
    q(x) = \max\left\{q\in\mathbb R: g\left(\mathcal C_q(x)\right)\le G_{\max}\right\}.
    \label{eq:constraint_rule_qx}
    \end{equation}
    \item \emph{Average operational constraint}: Given an average constraint level $G_{\max}$, we select the global threshold
    \begin{equation}
    q(\mathcal D) = \max\left\{q\in\mathbb R: \frac{1}{N_{\mathrm{cal}}}\sum_{i=1}^{N_{\mathrm{cal}}} g\left(\mathcal C_q(x_i)\right)\le G_{\max}\right\},
    \label{eq:constraint_rule_qd}
    \end{equation}
    where the average of the set functional is estimated using the calibration set $\mathcal D$. 
    The constraint~\eqref{eq:constraint_rule_qd} is more relaxed than~\eqref{eq:constraint_rule_qx}, since it does not require the upper bound $G_{\max}$ to be enforced for every input but only on average.
\end{enumerate}

Given a data-dependent threshold $q(x)$ or $q(\mathcal{D})$, the resulting set 
\begin{subequations}
\label{eq:threshold_set_comb}
\begin{align}
\label{eq:threshold_set_qx}
\mathcal C(x)= &
    \left\{
    y\in\mathcal Y:
    s(x,y)< q(x)
    \right\}\\
\text{or}\qquad \mathcal C(x)= &
    \left\{
    y\in\mathcal Y:
    s(x,y)< q(\mathcal{D})
    \right\}, 
\label{eq:threshold_set_qd}
\end{align} 
\end{subequations}
cannot generally satisfy a pre-defined miscoverage probability as in \eqref{eq:cp_marginal_cvg}.
In fact, guaranteeing the condition \eqref{eq:cp_marginal_cvg} may require larger sets than allowed by the thresholds $q(x)$ and $q(\mathcal{D})$, which are selected to satisfy the operational constraint in \eqref{eq:constraint_rule_qx} and \eqref{eq:constraint_rule_qd}, respectively~\cite{koning2025fuzzy}. In order to support reliable decision-making
based on set $\mathcal{C}(x)$ in \eqref{eq:threshold_set_comb}, it is critical to quantify the reliability of set $\mathcal{C}(x)$. To this end, one should determine a non-trivial upper bound $\hat{\alpha}(x)$ as a function of test input $x$ for the input-dependent constraint~\eqref{eq:constraint_rule_qx}, or a bound $\hat{\alpha}(\mathcal{D})$ as a function of the calibration data $\mathcal{D}$ for the average constraint~\eqref{eq:constraint_rule_qd}. Specifically, these estimates should ideally provide bounds on the conditional miscoverage rate as 
\begin{subequations}
\label{eq:rel_comb}
\begin{align}
\mathrm{Pr} \left(y\notin \mathcal{C}(x)|x\right) &\le \hat{\alpha}(x),
\label{eq:rel_poitwise_per_input}\\
\text{and}\qquad \mathrm{Pr} \left(y\notin \mathcal{C}(x)|\mathcal{D}\right) &\le \hat{\alpha}(\mathcal{D}), 
\label{eq:rel_poitwise_average}
\end{align} 
\end{subequations}
respectively.
With this information, the decision-maker could decide whether set $\mathcal{C}(x)$ is sufficiently reliable to be used, e.g., by setting a threshold on the estimate $\hat{\alpha}(x)$ or $\hat{\alpha}(\mathcal{D})$.

\subsection{Problem Definition}
In this paper, our problem consists in obtaining estimates of the miscoverage probabilities $\hat{\alpha}(x)$ or $\hat{\alpha}(\mathcal{D})$ and applying them to relevant wireless communication scenarios. However, obtaining
such estimates that satisfy the pointwise guarantees
in~\eqref{eq:rel_comb} is generally impossible \cite{koning2026right}. Instead, one can consider weaker, but still meaningful, requirements, such as the average-ratio criterion studied in \cite{gauthier2025values,gauthier2025backward,koning2023post}, and the probably approximately correct (PAC) criterion investigated in \cite{sarkar2023post}.
In both cases, we allow the estimate of the miscoverage rate to depend also on the calibration data $\mathcal{D}$, and we thus adopt the notation $\hat{\alpha}(x,\mathcal{D})$ under the input-dependent constraint in~\eqref{eq:constraint_rule_qx} and retain the notation $\hat{\alpha}(\mathcal D)$ under the average constraint in~\eqref{eq:constraint_rule_qd}.

\subsubsection{Average-Ratio Criterion}
\label{subsubsec:ratio_objective}
Rather than requiring the estimate $\hat{\alpha}(x,\mathcal D)$ to upper-bound
the miscoverage pointwise as in \eqref{eq:rel_poitwise_per_input}, the average-ratio criterion requires the reported level $\hat{\alpha}(x,\mathcal D)$ not to underestimate the
miscoverage on average \cite{koning2023post}. Mathematically, this is expressed as
\begin{equation}
    \mathbb E
    \left[
    \frac{
    \mathrm{Pr}
    \left(
    y
    \notin
    \mathcal C(x)|x,\mathcal D
    \right)
    }{
    \hat{\alpha}(x,\mathcal D)
    }
    \right]
    =\mathbb E
    \left[
    \frac{
    \mathbbm{1}
    \left(
    y
    \notin
    \mathcal C(x)
    \right)
    }{
    \hat{\alpha}(x,\mathcal D)
    }
    \right]
    \le
    1,
    \label{eq:cvg_guar_e}
\end{equation}
where the expectation on the left-hand side is taken over both the calibration data $\mathcal D$
and the test input $x$, while the expectation on the right-hand side also averages over the test label $y$, and $\mathbbm{1}(\cdot)$ denotes the indicator
function. Intuitively, condition~\eqref{eq:cvg_guar_e} penalizes underestimation: a smaller estimate $\hat{\alpha}(x,\mathcal{D})$ incurs a larger penalty whenever a miscoverage event $\mathbbm{1} \left(y\notin \mathcal{C}(x)\right)$ occurs. The criterion \eqref{eq:cvg_guar_e}
was shown in \cite{koning2026right} to satisfy a number of natural axioms for post-hoc inference.

\subsubsection{PAC-Based Criterion}
\label{subsubsec:pac_objective}

Focusing on estimates $\hat{\alpha} (\mathcal{D})$ relying
solely on the calibration data $\mathcal{D}$ as in \eqref{eq:rel_poitwise_average}, the
PAC criterion fixes a failure level $\delta \in(0,1)$, and imposes the requirement that the probability of an unreliable estimate $\hat{\alpha}(\mathcal{D})$ does not exceed $\delta$.
Mathematically, this criterion is formulated as the inequality~\cite{sarkar2023post}
\begin{equation}
    \mathrm{Pr}
    \left(
    \mathrm{Pr}
    \left(
    y
    \notin
    \mathcal C(x)
    \,\middle|\,
    \mathcal D
    \right)
    \le
    \hat{\alpha}(\mathcal D)
    \right)
    \ge
    1-\delta,
    \label{eq:pac_guar}
\end{equation}
where the inner probability is taken over the randomness of a test pair $(x,y)\sim p(x,y)$ conditional on the calibration set $\mathcal D$, and the outer probability is taken over the calibration data $\mathcal D$. Intuitively, across repeated draws of the calibration set, the reported miscoverage level
$\hat{\alpha}(\mathcal D)$ upper-bounds the true miscoverage
probability in at least $100(1-\delta)\%$ of the experiments.


\section{Miscoverage Estimation for Operationally Constrained Prediction Sets}
\label{sec:bcp_proc}
In this section, we propose the procedures for constructing reliable miscoverage estimates for the operationally constrained prediction sets induced by the input-dependent rule~\eqref{eq:constraint_rule_qx} or the average rule~\eqref{eq:constraint_rule_qd}.
Specifically,
we prove that BCP~\cite{gauthier2025backward} yields the input-specific estimate
$\hat{\alpha}(x,\mathcal D)$ satisfying the average-ratio objective~\eqref{eq:cvg_guar_e},
whereas $\delta$-PAC CP~\cite{sarkar2023post} yields the estimate
$\hat{\alpha}(\mathcal D)$ satisfying the PAC condition~\eqref{eq:pac_guar}. We start by reviewing a na\"{\i}ve baseline based
directly on the predictive distribution $\hat{p}(y|x)$.

\subsection{Na\"{\i}ve Miscoverage Estimate}
\label{subsec:naive}

Given a probabilistic predictor $\hat{p}(y|x)$ and a prediction set $\mathcal{C}(x)$, a straightforward way to estimate
the miscoverage level of the set is to use the predictive probability mass assigned
outside the set.
This yields the na\"{\i}ve
miscoverage estimate (NME)
\begin{equation}
    \hat{\alpha}^{\mathrm{NME}}(x)
    =
    1-\sum_{y \in \mathcal{C}(x)} \hat{p}(y|x).
    \label{eq:naive_alpha_discrete}
\end{equation}
For the singleton set $\mathcal C(x)=\{\hat y_x^{\mathrm{MAP}}\}$ induced by the MAP estimate in~\eqref{eq:point_est}, the NME reduces to the model-reported
misclassification probability
$1-\hat p(\hat y_x^{\mathrm{MAP}}|x)$. 
Similar expressions apply for continuous output spaces by replacing the sum with an integral.
The NME in \eqref{eq:naive_alpha_discrete} is generally not reliable in the sense of
\eqref{eq:cvg_guar_e} and \eqref{eq:pac_guar}, since it
directly relies on the confidence values reported by the predictive model.
In particular, when the model is overconfident~\cite{guo2017calibration}, the NME
tends to underestimate the true miscoverage \cite{romano2020classification}. This
limitation motivates the reliable miscoverage estimation methods developed next.


\subsection{Backward Conformal Prediction}
\label{subsec:bcp}
Consider the input-dependent threshold $q(x)$ in \eqref{eq:constraint_rule_qx} and assume that the score function $s(x,y)$ is non-negative, e.g., the IP score \eqref{eq:nc_score_power} and the LL score \eqref{eq:nc_score_log}.
Given the calibration scores $\{s(x_i,y_i)\}_{i=1}^{N_{\mathrm{cal}}}$ evaluated on the calibration dataset $\mathcal{D}$, BCP evaluates the estimate~\cite{gauthier2025backward,su2026reliable}
\begin{equation}
    \hat{\alpha}^{\mathrm{BCP}}(x,\mathcal D)
    =
    \frac{
        \sum_{i=1}^{N_{\mathrm{cal}}} s(x_i,y_i)
        +
        q(x)
    }{
        (N_{\mathrm{cal}}+1)q (x)
    }.
    \label{eq:alpha_hat_closed_form}
\end{equation}
Intuitively, when the prediction set in~\eqref{eq:score_threshold_set} is more confident, in the sense that the excluded labels have larger scores, the threshold $q(x)$ selected by \eqref{eq:constraint_rule_qx} increases, thus leading to a smaller miscoverage estimate $\hat{\alpha}^{\mathrm{BCP}}(x,\mathcal{D})$ in~\eqref{eq:alpha_hat_closed_form}.

\begin{proposition}[BCP reliability guarantee]
\label{prop:bcp_closed_form_validity}
For the prediction set $\mathcal C(x)$ in \eqref{eq:threshold_set_qx} selected by
\eqref{eq:constraint_rule_qx}, the BCP miscoverage estimate \eqref{eq:alpha_hat_closed_form} 
satisfies the reliability condition in \eqref{eq:cvg_guar_e}, i.e.,
\begin{equation}
    \mathbb{E}
    \left[
    \frac{
        \mathbbm{1}
        \left(
        y
        \notin
        \mathcal C(x)
        \right)
    }{
        \hat{\alpha}^{\mathrm{BCP}}(x,\mathcal D)
    }
    \right]
    \le
    1.
    \label{eq:post_hoc}
\end{equation}
\end{proposition}
\begin{IEEEproof}
See Appendix~\ref{app:bcp_proof}.    
\end{IEEEproof}

\subsection{Score-Transformation Backward Conformal Prediction} \label{subsec:st_bcp}
The BCP estimate $\hat{\alpha}^{\mathrm{BCP}}(x,\mathcal D)$ in~\eqref{eq:alpha_hat_closed_form} depends on the calibration scores through their magnitudes, and is thus generally sensitive to the choice of score function and can be overly conservative when the predictive model is less confident.
To mitigate the potential over-conservativeness of this solution, score-transformation BCP (ST-BCP)~\cite{liu2026improving} adopts a modified score evaluated through the monotonically non-decreasing transformation
\begin{equation}
    \tilde{s}(x,y) = q(x)\mathbbm{1}\{s(x,y)\geq q(x)\}.
    \label{eq:st_score}
\end{equation}
The transformation~\eqref{eq:st_score} collapses each score to one of two values, mapping the scores of the labels included in set~\eqref{eq:threshold_set_qx}
to zero and the remaining scores to $q(x)$. Accordingly, the rule \eqref{eq:constraint_rule_qx} yields the same threshold $q(x)$ as BCP, and hence  it preserves the BCP prediction set $\mathcal C(x)$ in \eqref{eq:threshold_set_qx}.

Replacing the calibration scores $\{s(x_i,y_i)\}_{i=1}^{N_{\mathrm{cal}}}$ in \eqref{eq:alpha_hat_closed_form} with the transformed scores $\tilde{s}(x_i,y_i)$ yields the ST-BCP estimate~\cite{liu2026improving}
\begin{equation}
    \hat{\alpha}^{\mathrm{ST\text{-}BCP}}(x,\mathcal D) = \frac{\sum_{i=1}^{N_{\mathrm{cal}}} q(x_i)\mathbbm{1}\{s(x_i,y_i)\geq q(x_i)\}+q(x)}{(N_{\mathrm{cal}}+1)q(x)}.
    \label{eq:st_alpha}
\end{equation}

\begin{corollary}[ST-BCP reliability guarantee]
\label{cor:st_reliability}
For the prediction set $\mathcal C(x)$ in \eqref{eq:threshold_set_qx} selected by \eqref{eq:constraint_rule_qx}, the ST-BCP miscoverage estimate \eqref{eq:st_alpha} satisfies the reliability condition in \eqref{eq:cvg_guar_e}, i.e.,
\begin{equation}
    \mathbb E\left[\frac{\mathbbm{1}\left(y\notin\mathcal C(x)\right)}{\hat{\alpha}^{\mathrm{ST\text{-}BCP}}(x,\mathcal D)}\right]\leq 1.
    \label{eq:st_reliability}
\end{equation}
Furthermore, the ST-BCP estimate~\eqref{eq:st_alpha} is no larger than the BCP estimate~\eqref{eq:alpha_hat_closed_form}, i.e.,
\begin{equation}
    \hat{\alpha}^{\mathrm{ST\text{-}BCP}}(x,\mathcal D) \leq \hat{\alpha}^{\mathrm{BCP}}(x,\mathcal D).
    \label{eq:st_bcp_ordering}
\end{equation}
\end{corollary}
\begin{IEEEproof}
Since the transformation preserves both the threshold $q(x)$ and the prediction set $\mathcal C(x)$, applying Proposition~\ref{prop:bcp_closed_form_validity} with the score function $s(\cdot,\cdot)$ replaced by $\tilde{s}(\cdot,\cdot)$ directly yields \eqref{eq:st_reliability}. Moreover, we have the inequality $\tilde{s}(x_i,y_i)\leq s(x_i,y_i)$ for every $i=1,\ldots,N_{\mathrm{cal}}$ by \eqref{eq:st_score}. Comparing \eqref{eq:st_alpha} with \eqref{eq:alpha_hat_closed_form} therefore yields \eqref{eq:st_bcp_ordering}.
\end{IEEEproof}

The ordering \eqref{eq:st_bcp_ordering} shows that ST-BCP generally yields smaller miscoverage estimates than BCP. However, a smaller estimate is not necessarily more accurate, and the comparison between BCP and ST-BCP must be addressed on a case-by-case basis.




\subsection{$\delta$-PAC Conformal Prediction}
\label{subsec:pac_cp}

We now consider the threshold $q(\mathcal D)$ selected according to the average operational constraint
\eqref{eq:constraint_rule_qd}, with the aim of satisfying the PAC criterion
\eqref{eq:pac_guar}. 
Let
$F(q)=\mathrm{Pr}(s(x,y)< q)$ denote the  CDF of the score. For any
deterministic threshold $q$, the coverage probability of the prediction set
in \eqref{eq:score_threshold_set} equals the CDF
\begin{equation}
    \mathrm{Pr}
    \left(
    y
    \in
    \mathcal C_q(x)
    \right)
    =
    F(q).
    \label{eq:coverage_score_cdf}
\end{equation}

Based on the calibration data $\mathcal D$, $\delta$-PAC CP \cite{sarkar2023post}
constructs a lower bound
$\ell_{\delta}(q;\mathcal D)$
for the CDF $F(q)$ that holds simultaneously over all thresholds $q\in\mathbb R$ with
probability at least $1-\delta$, i.e.,
\begin{equation}
    \mathrm{Pr}
    \left(
    \ell_{\delta}(q;\mathcal D)
    \le
    F(q),
    \forall q\in\mathbb R
    \right)
    \ge
    1-\delta,
    \label{eq:cdf_band}
\end{equation}
where the probability is taken over calibration data $\mathcal{D}$. For example, one can use the Dvoretzky--Kiefer--Wolfowitz (DKW)
inequality~\cite{massart1990tight} to construct the fixed-width lower confidence
bound in \eqref{eq:cdf_band} as
\begin{equation}
\ell_\delta^{\mathrm{DKW}}(q;\mathcal D)
    =
    \max\left\{
    \hat F(q;\mathcal D)
    -
    \sqrt{\frac{\log(2/\delta)}{2N_{\mathrm{cal}}}},
    0
    \right\},
    \label{eq:dkw_lower_band}    
\end{equation}
where 
\begin{equation}
    \hat F(q;\mathcal D)
    =
    \frac{1}{N_{\mathrm{cal}}}
    \sum_{i=1}^{N_{\mathrm{cal}}}
    \mathbbm{1}
    \left(
    s(x_i,y_i)< q
    \right)
    \label{eq:empirical_cdf}
\end{equation}
denotes the empirical CDF (ECDF) of the calibration scores.

In this paper, we leverage a sharper variable-width lower confidence bound in
\eqref{eq:cdf_band} using the method of D\"umbgen and Wellner (DW)
\cite{dumbgen2023new}, as detailed in Appendix~\ref{app:dw_band}. 

The estimate 
$\hat{\alpha}^{\mathrm{PAC}}(\mathcal D)$ is then obtained as 
\begin{equation}
    \hat{\alpha}^{\mathrm{PAC}}(\mathcal D)
    =
    1
    -
    \ell_{\delta}
    \left(
     q(\mathcal D);
    \mathcal D
    \right).
    \label{eq:pac_alpha_hat}
\end{equation}
Intuitively, the lower confidence band $\ell_{\delta}(q(\mathcal D);\mathcal D)$ assesses the coverage of the selected prediction set by examining how frequently the calibration scores are likely to fall below the threshold $q(\mathcal D)$. 

\begin{proposition}[PAC reliability guarantee]
\label{prop:pac_miscoverage_estimate}
For the prediction set $\mathcal C(x)$ in \eqref{eq:threshold_set_qd} selected by
\eqref{eq:constraint_rule_qd}, the PAC miscoverage estimate \eqref{eq:pac_alpha_hat} satisfies the PAC reliability
condition in \eqref{eq:pac_guar} for any fixed failure level $\delta\in(0,1)$, i.e.,
\begin{equation}
    \mathrm{Pr}
    \left(
    \mathrm{Pr}
    \left(
    y
    \notin
    \mathcal C(x)
    \middle|
    \mathcal D
    \right)
    \le
    \hat{\alpha}^{\mathrm{PAC}}(\mathcal D)
    \right)
    \ge
    1-\delta.
    \label{eq:pac_reliability}
\end{equation}
\end{proposition}

\begin{IEEEproof}
Using \eqref{eq:coverage_score_cdf}, for any deterministic threshold $q$, the
miscoverage probability is given by
$\mathrm{Pr}\left(y\notin\mathcal C_q(x)\right)=1-F(q)$. Therefore, the
confidence-band guarantee in \eqref{eq:cdf_band} implies
\begin{equation}
    \mathrm{Pr}
    \left(
    \mathrm{Pr}
    \left(
    y
    \notin
    \mathcal C_q(x)
    \right)
    \le
    1-\ell_{\delta}(q;\mathcal D),
    \quad
    \forall q\in\mathbb R
    \right)
    \ge
    1-\delta.
    \label{eq:simultaneous_miscov_bound}
\end{equation}
Since \eqref{eq:simultaneous_miscov_bound} holds uniformly over
$q\in\mathbb R$, evaluating it at the calibration-dependent threshold
$q=q(\mathcal D)$ yields the estimate
$\hat{\alpha}^{\mathrm{PAC}}(\mathcal D)$ in \eqref{eq:pac_alpha_hat}, which satisfies the PAC reliability condition in \eqref{eq:pac_reliability}.
\end{IEEEproof}

A related approach~\cite{zhou2025calibrating} can also satisfy the PAC criterion~\eqref{eq:pac_guar},
but its miscoverage estimate is generally more conservative than
\eqref{eq:pac_alpha_hat} with the DKW bound~\eqref{eq:dkw_lower_band}; we therefore defer details
to Appendix~C.

\begin{figure*}[!t]
    \centering

    \begin{minipage}[t]{\textwidth}
        \centering

        \begin{subfigure}[t]{0.48\linewidth}
            \centering
            \includegraphics[width=\linewidth]
            {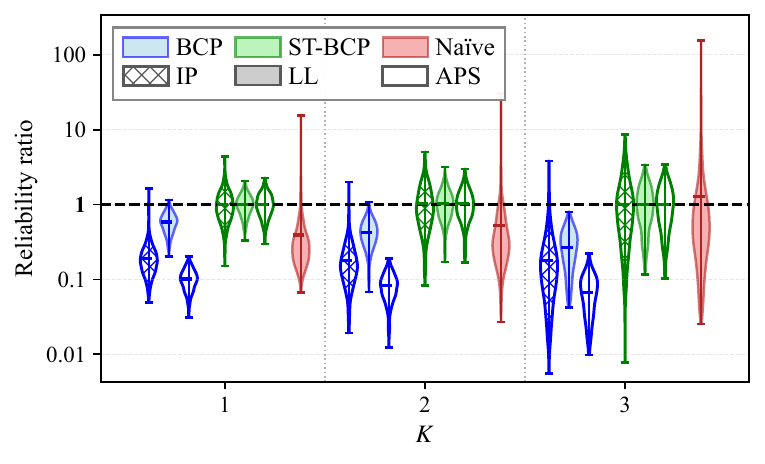}
            \caption{$\mathrm{SIR}=0$ dB.}
            \label{fig:intf_ratio_sir0_equal}
        \end{subfigure}
        \hfill
        \begin{subfigure}[t]{0.48\linewidth}
            \centering
            \includegraphics[width=\linewidth]
            {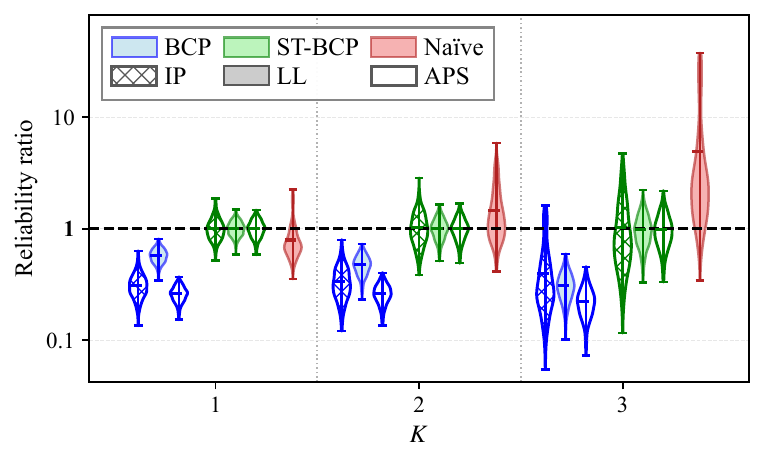}
            \caption{$\mathrm{SIR}=10$ dB.}
            \label{fig:intf_ratio_sir10_equal}
        \end{subfigure}

        \addtocounter{figure}{-1}
        \captionof{figure}{
        Violin plot of the reliability ratio~\eqref{eq:metric_rr}
        versus the budget $K$ for different NC scores under the
        equal-cost setting. Each violin shows the distribution over
        the $N_{\mathrm{run}}$ experiments, with the dashed line
        indicating the corresponding average.}
        \label{fig:intf_reliability_ratio_equal}

    \end{minipage}
    \par
    \begin{minipage}[t]{\textwidth}
        \centering

        \begin{subfigure}[t]{0.48\linewidth}
            \centering
            \includegraphics[width=\linewidth]
            {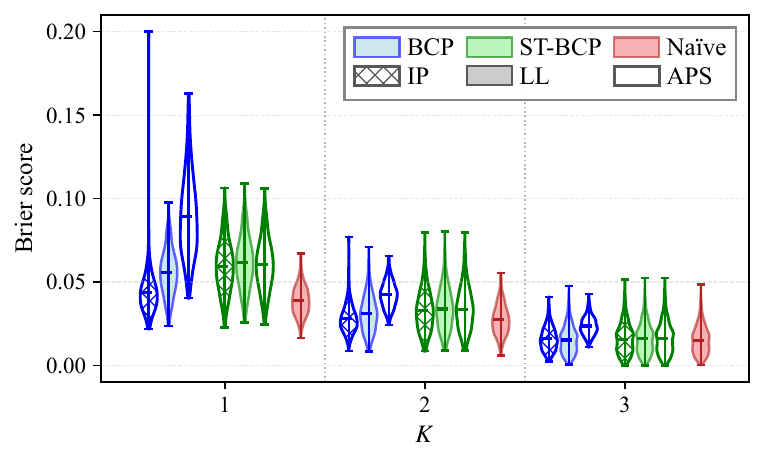}
            \caption{$\mathrm{SIR}=0$ dB.}
            \label{fig:intf_brier_sir0_equal}
        \end{subfigure}
        \hfill
        \begin{subfigure}[t]{0.48\linewidth}
            \centering
            \includegraphics[width=\linewidth]
            {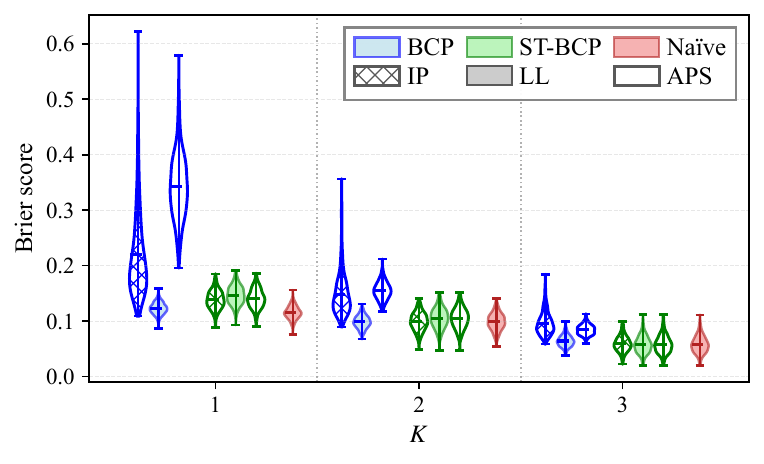}
            \caption{$\mathrm{SIR}=10$ dB.}
            \label{fig:intf_brier_sir10_equal}
        \end{subfigure}
        \captionof{figure}{
        Violin plot of the Brier score~\eqref{eq:metric_bs}
        versus the budget $K$ for different NC scores under the
        equal-cost setting.}
        \label{fig:intf_brier_score_equal}

    \end{minipage}

\end{figure*}

\section{Narrowband Interference Detection}
\label{sec:intf_det}
In this section, we present the first application of the miscoverage
estimation methods of Sec.~\ref{sec:bcp_proc}, namely budget-aware narrowband
interference detection~\cite{robinson2023narrowband}. This task
adopts the input-dependent operational constraint~\eqref{eq:constraint_rule_qx} and is
addressed using BCP (Sec.~\ref{subsec:bcp}) and
ST-BCP (Sec.~\ref{subsec:st_bcp}).

\begin{figure*}[!t]
    \centering

    \begin{minipage}[t]{\textwidth}
        \centering
        \begin{subfigure}[t]{0.49\linewidth}
            \centering
            \includegraphics[width=\linewidth]{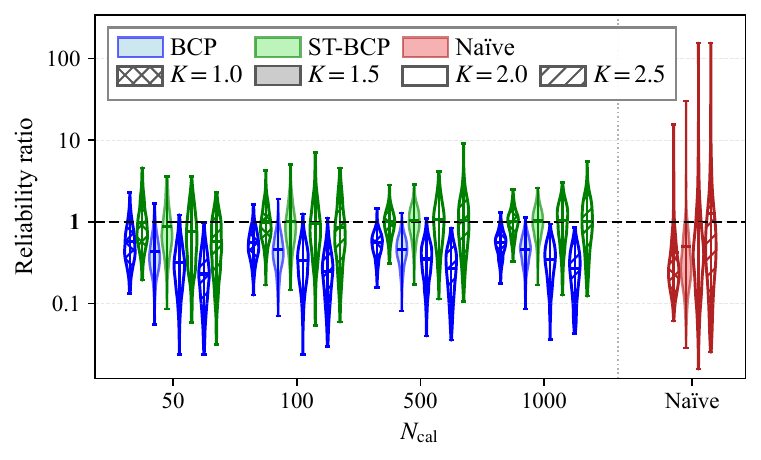}
            \caption{$\mathrm{SIR}=0$ dB.}
            \label{fig:intf_ratio_sir0_unequal}
        \end{subfigure}
        \hfill
        \begin{subfigure}[t]{0.49\linewidth}
            \centering
            \includegraphics[width=\linewidth]{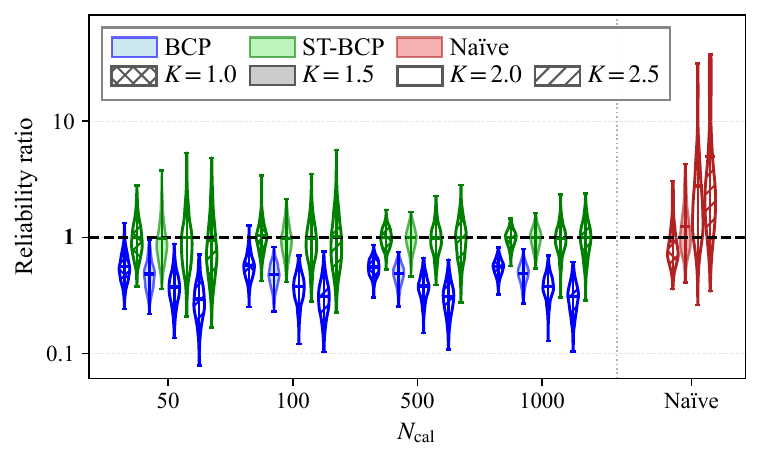}
            \caption{$\mathrm{SIR}=10$ dB.}
            \label{fig:intf_ratio_sir10_unequal}
        \end{subfigure}

        \addtocounter{figure}{-1}
        \captionof{figure}{Violin plot of the reliability ratio~\eqref{eq:metric_rr} versus the calibration size $N_{\mathrm{cal}}$ for different budgets $K$ under the unequal-cost setting.}
        \label{fig:intf_reliability_ratio_unequal}
    \end{minipage}
    \par
    \begin{minipage}[t]{\textwidth}
        \centering
        \begin{subfigure}[t]{0.49\linewidth}
            \centering
            \includegraphics[width=\linewidth]{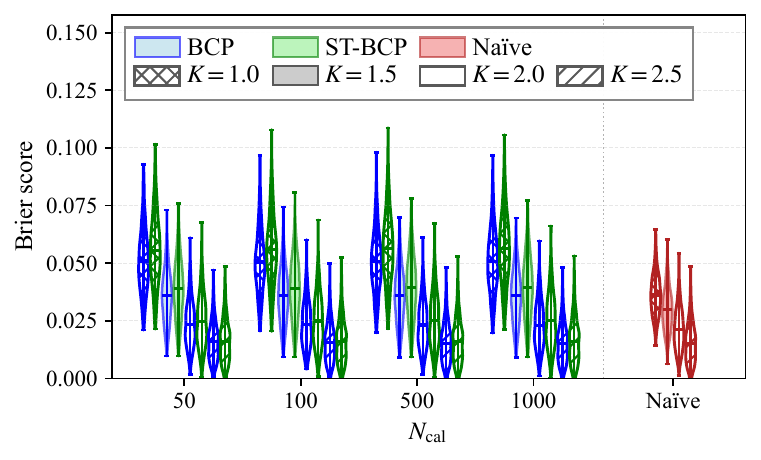}
            \caption{$\mathrm{SIR}=0$ dB.}
            \label{fig:intf_brier_sir0_unequal}
        \end{subfigure}
        \hfill
        \begin{subfigure}[t]{0.49\linewidth}
            \centering
            \includegraphics[width=\linewidth]{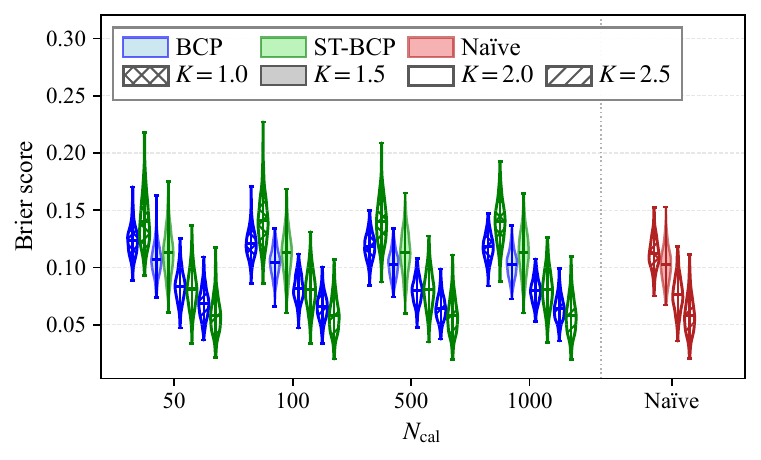}
            \caption{$\mathrm{SIR}=10$ dB.}
            \label{fig:intf_brier_sir10_unequal}
        \end{subfigure}

        \captionof{figure}{Violin plot of the Brier score~\eqref{eq:metric_bs} versus the calibration size $N_{\mathrm{cal}}$ for different budgets $K$ under the unequal-cost setting.}
        \label{fig:intf_brier_score_unequal}
    \end{minipage}

\end{figure*}

\subsection{Problem Formulation}
We consider a wideband system operating in the presence of NBI~\cite{robinson2023narrowband}.
Following~\cite{robinson2023narrowband}, the receiver collects baseband I/Q samples $x$ and determine whether NBI is present and, if so, identify the affected subcarriers. We focus on NBI-affected observations and monitor $S$ subcarriers, assuming, as in~\cite{robinson2023narrowband}, that the NBI affects only one of them.
The output $y$ identifies the affected subcarrier, with label space $\mathcal{Y}=\{1,\ldots,S\}$.


In this application, we consider the input-dependent constraint in~\eqref{eq:constraint_rule_qx}, since the budget $K\in\mathbb{R}$ available for downstream mitigation may be limited by the receiver hardware or per-input performance requirements~\cite{same2021multiple,bedeer2012adaptive}.
To model this constraint, we use the weighted set-size functional
in~\eqref{eq:weighted_cardinality}, where $w_y\in(0,1]$ denotes the
normalized cost of mitigating subcarrier $y$. For example, subcarriers carrying more bits can be assigned larger mitigation cost $w_y$, since including them in the set incurs a larger throughput loss~\cite{bedeer2012adaptive}.
Our goal is then to
construct, for the resulting set $\mathcal C(x)$
in~\eqref{eq:threshold_set_qx}, a miscoverage estimate that satisfies
the reliability condition in~\eqref{eq:cvg_guar_e}.


\subsection{BCP-Based Solution}
\label{subsec:bcp_intf}

We use the probabilistic classifier in
\cite{robinson2023narrowband} as the interference detector to obtain the predictive distribution
$\hat{p}(y|x)$ over the $S$ candidate subcarriers. To construct the prediction set, we adopt the LL score in
\eqref{eq:nc_score_log}. 
Let $y^{(s)}$ denote the $s$-th most likely output under the predictive distribution $\hat{p}(y|x)$, with ties broken arbitrarily.
Applying the threshold-selection rule in
\eqref{eq:constraint_rule_qx} with $G_{\max}=K$, we obtain
\begin{equation}
    q(x)
    =
    -\log \hat{p}(
    y^{(r_K(x)+1)}
    |
    x),
    \label{eq:q_intf}
\end{equation}
where $r_K(x)
    =
    \max\{
    r\in\{1,\ldots,S\}:
    \sum_{s=1}^{r}w_{y^{(s)}}\le K
    \}$.
Thus, $q(x)$ in \eqref{eq:q_intf} corresponds to the score of the most likely label excluded from
the prediction set $\mathcal{C}(x)$. Given the calibration data $\mathcal{D}$, evaluating
\eqref{eq:alpha_hat_closed_form} and~\eqref{eq:st_alpha} using the
threshold in~\eqref{eq:q_intf} yields the miscoverage estimates
$\hat{\alpha}^{\mathrm{BCP}}(x,\mathcal D)$ and
$\hat{\alpha}^{\mathrm{ST\text{-}BCP}}(x,\mathcal D)$, respectively,
both satisfying the average-ratio reliability guarantee
in~\eqref{eq:cvg_guar_e}.

\subsection{Numerical Results}

\subsubsection{Simulation Setup}
We investigate a WiFi system with narrowband interference, following the setup in~\cite{robinson2023narrowband}. In particular, the bandwidth of interest is partitioned into $S=4$ monitored subcarriers.
All WiFi I/Q data are generated according to the IEEE 802.11a standard using the MATLAB WLAN Toolbox. The resulting training dataset contains approximately $120{,}000$ examples per label, with the signal-to-interference ratio (SIR) ranging from $-10$ to $10$ dB. Using this training dataset, we implement the detector following the CNN architecture in~\cite{robinson2023narrowband}, except that stochastic gradient descent is adopted as the optimizer. For evaluation, we use two held-out datasets, each containing $3{,}000$
examples per label, collected at $\mathrm{SIR}=0$ dB and
$\mathrm{SIR}=10$ dB, respectively.
We compare the BCP and ST-BCP miscoverage estimates in Sec.~\ref{subsec:bcp_intf} with the NME in
Sec.~\ref{subsec:naive}.

\subsubsection{Performance Metrics}
For all considered schemes, performance is evaluated on the test dataset $\mathcal{D}^{\mathrm{te}}=\{(x_j,y_j)\}_{j=1}^{N_{\mathrm{te}}}$, with
$N_{\mathrm{te}}=200$, in terms of both reliability and accuracy. Let $\hat{\alpha}_j$ and $m_j=\mathbbm{1}(y_j\notin\mathcal{C}(x_j))$ denote, respectively, the miscoverage estimate and the corresponding true miscoverage event for the $j$-th test sample.

Based on these quantities, we report the
\begin{equation}
\text{Reliability ratio}
=
\frac{1}{N_{\mathrm{te}}}\sum_{j=1}^{N_{\mathrm{te}}}\frac{m_j}{\hat{\alpha}_j}
\label{eq:metric_rr}
\end{equation}
to assess reliability. The estimates $\hat{\alpha}_j$ are reliable in the sense of \eqref{eq:cvg_guar_e} if the
reliability ratio in \eqref{eq:metric_rr}, averaged over independent
experiments, is no greater than one. Furthermore, to quantify accuracy, we report the 
\begin{equation}
\label{eq:metric_bs}
    \text{Brier~score}
    =
    \frac{1}{N_{\mathrm{te}}}
    \sum_{j=1}^{N_{\mathrm{te}}}
    \left(m_j-\hat{\alpha}_j\right)^2,
\end{equation}
which measures the mean squared error between the per-input
estimated and true miscoverage.
We average these metrics over $N_{\mathrm{run}}=500$ independent experiments. In each experiment, disjoint calibration and test sets $\{\mathcal{D}^{\mathrm{cal}},\mathcal{D}^{\mathrm{te}}\}$ are randomly drawn from the held-out dataset.
Unless otherwise stated, all violin plots show the distribution over
the $N_{\mathrm{run}}$ experiments, with dashed lines indicating the
corresponding averages.

\subsubsection{Performance Analysis}
We first consider the \textit{equal-cost setting}, with $w_y=1$ for all $y\in\mathcal{Y}$ and $N_{\mathrm{cal}}=1000$. 
For BCP and ST-BCP, we adopt the IP score in~\eqref{eq:nc_score_power} with $\beta=0.75$, the LL score in~\eqref{eq:nc_score_log}, and the randomized adaptive prediction sets (APS)-based score~\cite{romano2020classification}. The parameter $\beta$ is selected from a discrete set of candidate values using a separate validation set.
In particular, using the same ordering $\{y^{(s)}\}_{s=1}^{S}$ as in Sec.~\ref{subsec:bcp_intf}, the randomized APS-based score is defined as
\begin{equation}
s(x,y^{(s)})
=
-\log\left(
1-
\sum_{j=1}^{s-1}\hat p(y^{(j)}|x)
-
U\hat p(y^{(s)}|x)
\right),
\label{eq:aps_score}
\end{equation}
where $U\sim\operatorname{Unif}[0,1]$ is drawn independently for each sample.

Fig.~\ref{fig:intf_reliability_ratio_equal} compares the reliability ratio~\eqref{eq:metric_rr} across different NC scores and budgets $K$. Under both SIR conditions, the average reliability ratios of BCP and ST-BCP remain no greater than one across all considered scores and budgets, which is consistent with the reliability guarantee in~\eqref{eq:cvg_guar_e}. In contrast, the reliability ratio of the NME frequently exceeds one, particularly at higher SIR. This is because weaker interference makes the detection task more challenging, while the 
NME remains overconfident.

The estimation accuracy is further assessed in Fig.~\ref{fig:intf_brier_score_equal} using the Brier score~\eqref{eq:metric_bs}. For BCP, the Brier score is noticeably affected by the choice of NC score, with the IP score generally yielding lower values and the randomized APS-based score higher values, particularly for smaller $K$. In contrast, ST-BCP is much less sensitive to the NC score, with largely overlapping Brier-score distributions across the three choices. With an appropriate NC score, both BCP and ST-BCP achieve Brier scores comparable to those of the NME, and the gap generally decreases as $K$ increases.

We next consider the more general \textit{unequal-cost setting} with $(w_1,w_2,w_3,w_4)=(1,1,0.5,0.5)$, and use the LL score in~\eqref{eq:nc_score_log} to investigate the impact of the calibration size $N_{\mathrm{cal}}$.
Fig.~\ref{fig:intf_reliability_ratio_unequal} shows that the reliability guarantee is preserved under unequal costs, with the average reliability ratios of BCP and ST-BCP remaining no greater than one across all considered calibration sizes and budgets. In contrast, the NME can violate the reliability condition, particularly at higher SIR and larger budgets.

Fig.~\ref{fig:intf_brier_score_unequal} further evaluates the accuracy under unequal costs. As $N_{\mathrm{cal}}$ increases, the Brier-score distributions of both BCP and ST-BCP become more concentrated, as the empirical averages of the calibration scores in~\eqref{eq:alpha_hat_closed_form} and~\eqref{eq:st_alpha} become more stable. This effect is more pronounced at $\mathrm{SIR}=10$ dB, where the estimates exhibit larger variability for small calibration sets. Overall, both methods achieve Brier scores close to those of the NME.

\section{Near-Field Localization}
\label{sec:loc}
In this section, we investigate
utility-aware near-field localization~\cite{mozaffarikhosravi2025localization}
under the input-dependent operational constraint in
\eqref{eq:constraint_rule_qx}. BCP is applied to obtain reliable miscoverage
estimates for prediction sets satisfying a prescribed beamforming-gain
requirement.

\subsection{Problem Formulation}
We consider near-field localization, in which an $N$-antenna BS estimates the location $y$ (i.e., the distance from the BS) of a single-antenna user based on the received pilot signal $x$~\cite{mozaffarikhosravi2025localization}. The user is located in the radiative near-field region, with distance $y\in[R_{\mathrm{FD}},R_{\mathrm{RD}}]$, where $R_{\mathrm{FD}}$ and $R_{\mathrm{RD}}$ denote the Fresnel and Rayleigh distances, respectively, as defined in~\cite{selvan2017fraunhofer}. Following~\cite{mozaffarikhosravi2025localization}, a point estimate of the user location $\hat y_x$ can be obtained, for a known user angle, by applying the multiple signal classification (MUSIC) algorithm to the input $x$.

In practice, the estimate  $\hat y_x$ need not be exact for
downstream beam focusing, provided that focusing the beam at location $\hat y_x$
delivers sufficient gain to a user at the true location
$y$~\cite{gavriilidis2025nearfield}. This tolerance is captured by the relative beamforming gain $u_x(y)$, defined as the array gain obtained when focusing at location $\hat y_x$ for a user at location $y$, normalized by the gain under perfect focusing~\cite{gavriilidis2025nearfield}, so that $u_x(\hat y_x)=1$. 
To maintain each user’s link quality under localization errors~\cite{sheemar2025joint}, we impose the input-dependent constraint in~\eqref{eq:constraint_rule_qx}, requiring the relative beamforming gain to exceed a prescribed level $\lambda\in(0,1)$~\cite{gavriilidis2025nearfield}.
To model this
constraint, we use the worst-case loss functional
in~\eqref{eq:worst_case_utility} with $w_y=1-u_x(y)$ and set
$G_{\max}=1-\lambda$ in~\eqref{eq:constraint_rule_qx}. 
Based on this formulation, our goal is to construct a reliable per-input miscoverage estimate $\hat{\alpha}(x,\mathcal D)$ for the set $\mathcal C(x)$~\eqref{eq:threshold_set_qx} under the constraint~\eqref{eq:constraint_rule_qx}, that satisfies the average-ratio criterion in~\eqref{eq:cvg_guar_e}.

\begin{figure*}[!t]
    \centering

    \begin{minipage}[t]{\textwidth}
        \centering
        \begin{subfigure}[t]{0.49\linewidth}
            \centering
            \includegraphics[width=\linewidth]{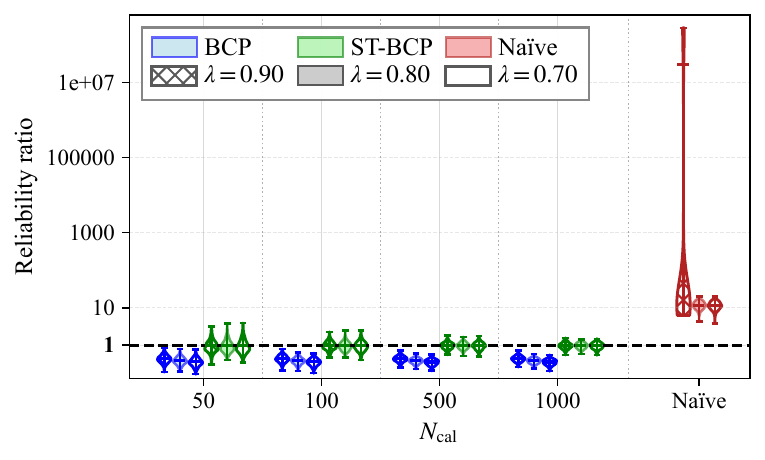}
            \caption{$\mathrm{SNR}=10$ dB.}
            \label{fig:ratio_loc_10db}
        \end{subfigure}
        \hfill
        \begin{subfigure}[t]{0.49\linewidth}
            \centering
            \includegraphics[width=\linewidth]{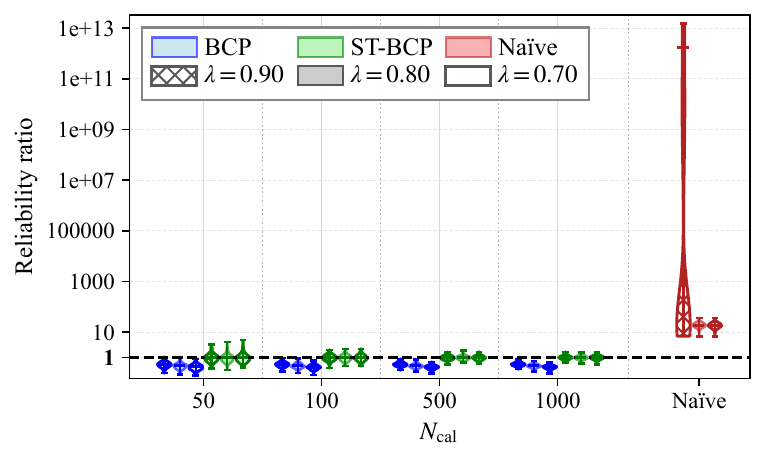}
            \caption{$\mathrm{SNR}=25$ dB.}
            \label{fig:ratio_loc_25db}
        \end{subfigure}

        \addtocounter{figure}{-1}
        \captionof{figure}{Violin plots of the reliability ratio~\eqref{eq:metric_rr}
        versus the calibration size $N_{\mathrm{cal}}$ for gain levels
        $\lambda\in\{0.9,0.8,0.7\}$.}
        \label{fig:ratio_loc}
    \end{minipage}
    \par
    \begin{minipage}[t]{\textwidth}
        \centering
        \begin{subfigure}[t]{0.49\linewidth}
            \centering
            \includegraphics[width=\linewidth]{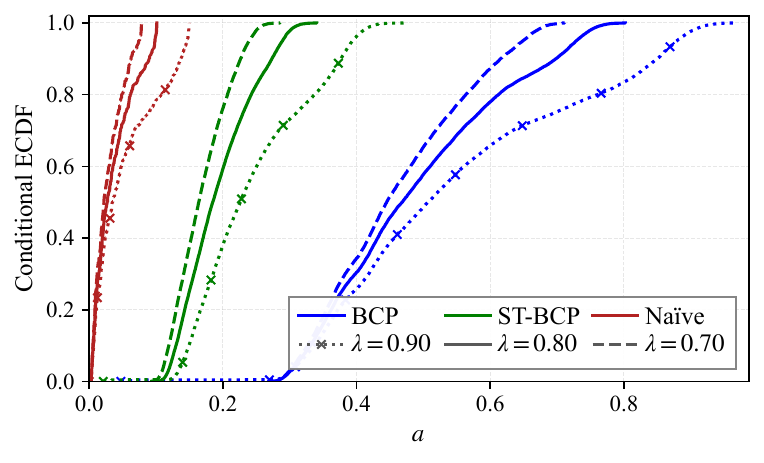}
            \caption{$\mathrm{SNR}=10$ dB.}
            \label{fig:ecdf_loc_10db}
        \end{subfigure}
        \hfill
        \begin{subfigure}[t]{0.49\linewidth}
            \centering
            \includegraphics[width=\linewidth]{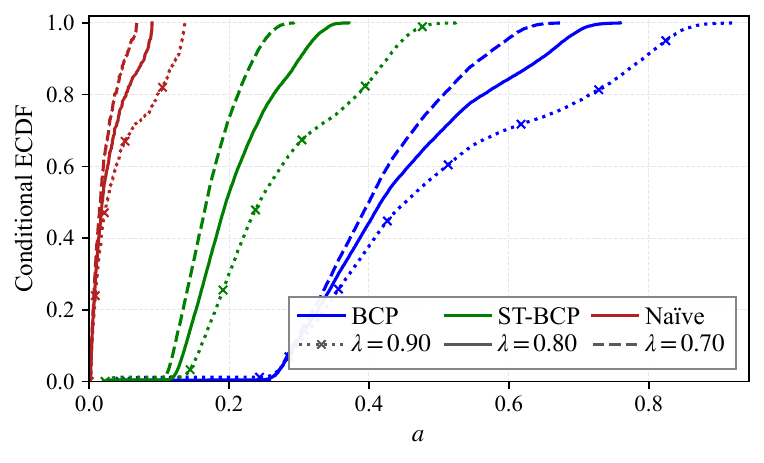}
            \caption{$\mathrm{SNR}=25$ dB.}
            \label{fig:ecdf_loc_25db}
        \end{subfigure}
        \captionof{figure}{Conditional ECDF~\eqref{eq:metric_cecdf} of the estimated miscoverage level $\hat{\alpha}$ given miscoverage for gain levels $\lambda\in\{0.9,0.8,0.7\}$. For BCP and ST-BCP, the calibration size is fixed to $N_{\mathrm{cal}}=1000$.}
        \label{fig:ecdf_loc}
    \end{minipage}

\end{figure*}

\subsection{BCP-Based Solution}
\label{subsec:bcp_loc}
To instantiate the prediction set $\mathcal C(x)$ in~\eqref{eq:threshold_set_qx},
we first construct the probability predictor as follows. Following~\cite{mozaffarikhosravi2025localization}, the conditional distribution of the distance estimate is approximated as $p(\hat y_x|y)=\mathcal G(1/(\eta y^2),\eta y^3)$, where $\mathcal G(\cdot,\cdot)$ denotes the Gamma distribution parameterized by its shape and scale, and $\eta>0$ is fitted using the training data, e.g., via least squares. Given a prior density $p(y)$, Bayes' theorem yields the posterior probability density function (PDF)
\begin{equation}
    p(y|\hat y_x)
    =
    \frac{p(\hat y_x|y)p(y)}
    {\int_{R_{\mathrm{FD}}}^{R_{\mathrm{RD}}} p(\hat y_x|z)p(z)\,dz},
    \quad
    y\in[R_{\mathrm{FD}},R_{\mathrm{RD}}],
    \label{eq:posterior_loc_general}
\end{equation}
and $p(y|\hat y_x)=0$ otherwise.

It remains to specify the score $s(x,y)$ introduced in~\eqref{eq:score_threshold_set}. 
Let
$\underline{u}_x(y)\triangleq\min_{z\in[\hat y_x,y]}u_x(z)$ denote the
minimum relative beamforming gain between distances $\hat y_x$ and $y$, where
$[a,b]$ denotes the interval with endpoints $a$ and $b$. Using the
posterior PDF in~\eqref{eq:posterior_loc_general}, we define the NC score
from the posterior mass of the corresponding gain-superlevel interval as
\begin{equation}
s(x,y)
=
-\log\left(
1-
\int_{\{z\in\mathcal Y:\,
\underline{u}_x(z)>\underline{u}_x(y)\}}
p(z\mid\hat y_x)\,dz
\right).
\label{eq:score_loc}
\end{equation}
Hence, a candidate distance $y$ closer to $\hat y_x$ is assigned a smaller score.
Using the score in~\eqref{eq:score_loc}, the score threshold
$q(x)$ in~\eqref{eq:constraint_rule_qx} corresponding to the target gain
level $\lambda$ is
\begin{equation}
q(x)
=
-\log\left(
1-
\int_{\{z\in\mathcal Y:\,\underline{u}_x(z)>\lambda\}}
p(z\mid\hat y_x)\,dz
\right).
\label{eq:threshold_loc}
\end{equation}
With the threshold in~\eqref{eq:threshold_loc}, the set $\mathcal C(x)$
in~\eqref{eq:threshold_set_qx} is the connected component of the
$\lambda$-superlevel set of $u_x$ that contains $\hat y_x$. 
Substituting the calibration scores
$\{s(x_i,y_i)\}_{i=1}^{N_{\mathrm{cal}}}$ computed on the calibration
data $\mathcal D$ and the threshold $q(x)$ in~\eqref{eq:threshold_loc}
into \eqref{eq:alpha_hat_closed_form} and~\eqref{eq:st_alpha} yields the miscoverage estimates
$\hat{\alpha}^{\mathrm{BCP}}(x,\mathcal D)$ and $\hat{\alpha}^{\mathrm{ST-BCP}}(x,\mathcal D)$, respectively, both satisfying the
reliability guarantee in~\eqref{eq:cvg_guar_e}.

\subsubsection{Simulation Setup}

We adopt the same near-field localization setup and parameters as
in~\cite{mozaffarikhosravi2025localization}.
At $\mathrm{SNR}\in\{10,25\}$ dB, we generate a dataset of $5000$ pairs $(\hat y_x,y)$ for each SNR value, with uniformly distributed user locations in the radiative near-field region. The dataset is equally split into a training subset for fitting the Gamma-model parameter $\eta$ and a held-out subset for calibration and testing.
For each target relative beamforming-gain level
$\lambda\in\{0.9,0.8,0.7\}$, we compare BCP and ST-BCP
in Sec.~\ref{subsec:bcp_loc} with the NME in
Sec.~\ref{subsec:naive}.
Performance is evaluated using the same metrics as in
Sec.~\ref{sec:intf_det}, namely, the reliability ratio
in~\eqref{eq:metric_rr} and the Brier score in~\eqref{eq:metric_bs}.
Results are averaged over $N_{\mathrm{run}}=500$ independent experiments. In each experiment, disjoint calibration and test sets $\{\mathcal{D}^{\mathrm{cal}},\mathcal{D}^{\mathrm{te}}\}$ are randomly drawn from the held-out dataset, with $N_{\mathrm{te}}=200$ test samples.

\subsubsection{Performance Analysis}
Fig.~\ref{fig:ratio_loc} compares the reliability ratios~\eqref{eq:metric_rr} of BCP,
ST-BCP, and the NME at $\mathrm{SNR}=10$ and $25$ dB. Across both SNR conditions and all considered gain levels and calibration
sizes, the average reliability ratios of both BCP and ST-BCP remain no
greater than one, empirically validating the reliability guarantee
in~\eqref{eq:cvg_guar_e}. Furthermore, the distributions of
both methods become increasingly concentrated as the calibration set $N_{\mathrm{cal}}$
grows. In contrast, the NME produces reliability ratios substantially above
one for all considered gain levels, with more severe violations at $\mathrm{SNR}=25$ dB.
The violation is particularly pronounced at $\lambda=0.9$, where the
distribution exhibits an extremely heavy upper tail because the narrower
intervals yield a higher miscoverage rate that is severely underestimated
by the NME.

To further investigate the extreme reliability ratios of the NME, let $\hat{\alpha}_{r,j}$ and
$m_{r,j}=\mathbbm{1}(y_{r,j}\notin\mathcal{C}(x_{r,j}))$ denote,
respectively, the miscoverage estimate and the corresponding true
miscoverage indicator for the $j$-th test sample in the $r$-th experiment. Fig.~\ref{fig:ecdf_loc} reports the
\begin{equation}
\text{Conditional ECDF}(a)
=
\frac{
\displaystyle\sum_{r=1}^{N_{\mathrm{run}}}
\sum_{j=1}^{N_{\mathrm{te}}}
m_{r,j}\mathbbm{1}\left(\hat{\alpha}_{r,j}\leq a\right)
}{
\displaystyle\sum_{r=1}^{N_{\mathrm{run}}}
\sum_{j=1}^{N_{\mathrm{te}}}m_{r,j}
},
\label{eq:metric_cecdf}
\end{equation}
which denotes the ECDF of the estimated miscoverage levels over all
miscovered test samples.
Across both SNR conditions and all considered values of $\lambda$, the conditional ECDFs of the
NME rise sharply near zero, indicating that a substantial fraction of
the observed miscoverage events are assigned very small miscoverage
estimates. At $\lambda=0.9$, the higher miscoverage rate, combined with this
mismatch, results in the extremely large reliability ratios and the
heavy upper tails observed in Fig.~\ref{fig:ratio_loc}.
By contrast, the absence of near-zero estimates under BCP and ST-BCP
prevents individual miscoverage events from dominating the reliability
ratio, as reflected by their average ratios remaining below one in
Fig.~\ref{fig:ratio_loc}.

\begin{figure*}[!t]
    \centering

    \begin{minipage}[t]{\textwidth}
        \centering
        \begin{subfigure}[t]{0.49\linewidth}
            \centering
            \includegraphics[width=\linewidth]
            {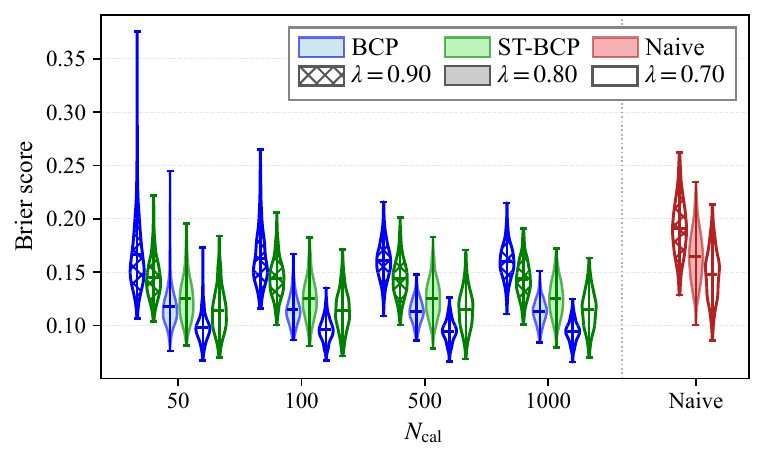}
            \caption{$\mathrm{SNR}=10$ dB.}
            \label{fig:brier_loc_10db}
        \end{subfigure}
        \hfill
        \begin{subfigure}[t]{0.49\linewidth}
            \centering
            \includegraphics[width=\linewidth]
            {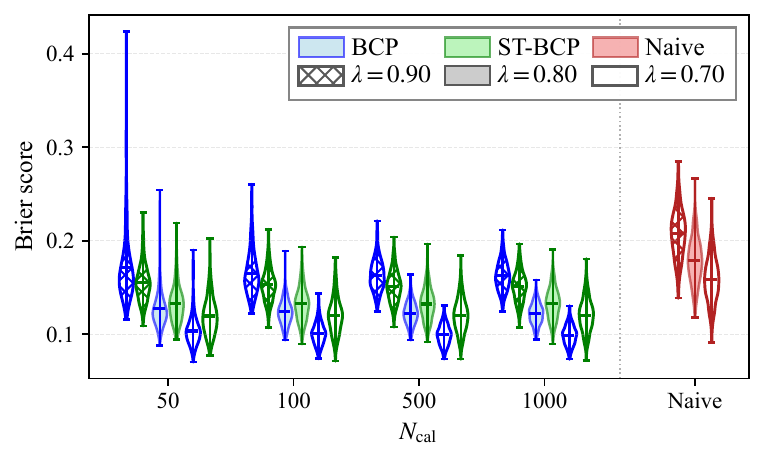}
            \caption{$\mathrm{SNR}=25$ dB.}
            \label{fig:brier_loc_25db}
        \end{subfigure}

        \addtocounter{figure}{-1}
        \captionof{figure}{Violin plots of the Brier score~\eqref{eq:metric_bs}
        versus the calibration size $N_{\mathrm{cal}}$ for levels
        $\lambda\in\{0.9,0.8,0.7\}$.}
        \label{fig:brier_loc}
    \end{minipage}
    \par
    \begin{minipage}[t]{\textwidth}
        \centering
        \begin{subfigure}[t]{0.49\linewidth}
            \centering
            \includegraphics[width=\linewidth]
            {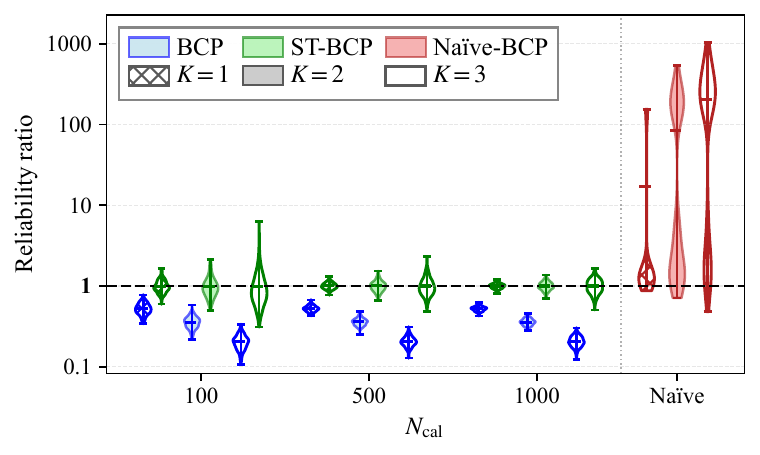}
            \caption{Reliability ratios of BCP, ST-BCP, and na\"{\i}ve-BCP.}
            \label{fig:bf_reliability_ratio}
        \end{subfigure}
        \hfill
        \begin{subfigure}[t]{0.49\linewidth}
            \centering
            \includegraphics[width=\linewidth]
            {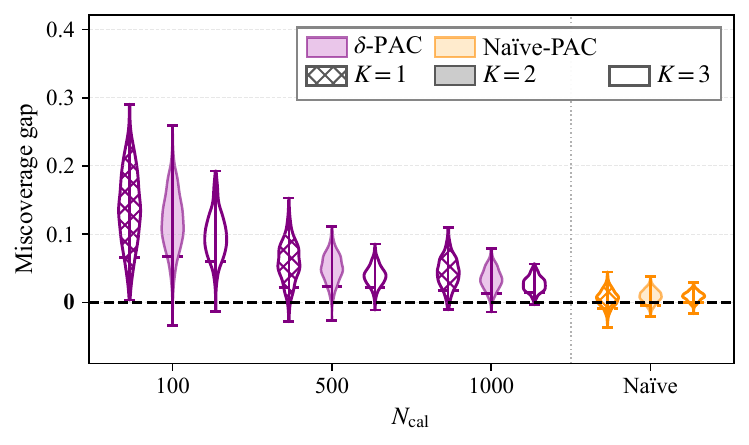}
            \caption{Miscoverage gaps of $\delta$-PAC CP and na\"{\i}ve-PAC.}
            \label{fig:bf_pac_gap}
        \end{subfigure}
        \captionof{figure}{Reliability performance versus the calibration size $N_{\mathrm{cal}}$
        for budgets $K\in\{1,2,3\}$. In (b), the dashed lines indicate the lower $\delta$-quantiles rather than the averages.}
        \label{fig:bf_reliability}
    \end{minipage}

\end{figure*}

Fig.~\ref{fig:brier_loc} evaluates the accuracy of the miscoverage
estimates in terms of the Brier score~\eqref{eq:metric_bs}. Across both SNR conditions and all considered gain levels $\lambda$, BCP and ST-BCP achieve comparable
average Brier scores, whereas the NME yields markedly higher scores due
to its systematic underestimation of the miscoverage level. Moreover, the Brier-score distributions of both BCP and ST-BCP become
more concentrated as $N_{\mathrm{cal}}$ increases, indicating more
stable accuracy across experiments.


\section{Codebook-Based Beam Identification}
\label{sec:bf}
In this section, we consider budget-aware codebook-based beam identification~\cite{othman2026diffusion}. 
Unlike the previous two applications, where the size constraint is imposed on each input separately, the probing budget here is imposed on average using \eqref{eq:constraint_rule_qd}. This setting is addressed by $\delta$-PAC CP (Sec.~\ref{subsec:pac_cp}), but we also provide a comparison with BCP (Sec.~\ref{subsec:bcp}) and
ST-BCP (Sec.~\ref{subsec:st_bcp}), which impose the stricter per-input set requirement~\eqref{eq:constraint_rule_qx}.

\subsection{Problem Formulation}
We address beam identification over a directional wireless channel using a codebook of $B$ predefined beams, where $y=b$, for $b\in\{1,\ldots,B\}$, denotes the $b$-th beam direction~\cite{othman2026diffusion}.
Given side information $x$ (i.e., the UE location), the BS aims to identify the optimal beam $y$ that yields the highest received power at the UE. 
The BS constructs a candidate beam set $\mathcal C(x)$ based on the predicted beam prior $\hat p(y| x)$, allowing it to probe the beams in this set using one pilot signal per beam to select the beam used for data transmission~\cite{orimogunje2026sensing}.

Probing a candidate beam consumes training and feedback resources, so we measure the probing overhead using the cardinality functional $g(\mathcal C(x))=|\mathcal C(x)|$. 
Here, we consider the case in which the probing overhead is budgeted jointly across multiple beam-identification instances and is therefore constrained on average~\cite{orimogunje2026sensing}.
We therefore impose the average operational constraint~\eqref{eq:constraint_rule_qd}, with an average budget of $K\in \mathbb{R}^{+}$ beams.


\subsection{BCP-Based Solution}
\label{subsec:bcp_bf}
The average constraint~\eqref{eq:constraint_rule_qd} is clearly met by imposing the stronger requirement that the set $\mathcal{C}(x)$ contain at most $K$ candidate beams, which can be addressed by BCP.

To obtain the predictive beam prior $\hat{p}(y|x)$ over the $B$ beams, we use the conditional diffusion model in~\cite{othman2026diffusion}. Let $y^{(b)}$ denote the beam with the $b$-th largest predictive probability, with ties broken arbitrarily. Using again the LL score in
\eqref{eq:nc_score_log}, the input-dependent constraint~\eqref{eq:constraint_rule_qx} with $G_{\max}=K$ then yields the threshold
\begin{equation}
    q(x)
    =
    -\log \hat{p}(
    y^{(K+1)}
    |
    x),
    \label{eq:q_bf_bcp}
\end{equation}
which is the score of the $(K+1)$-th most probable beam for input $x$. 
Evaluating~\eqref{eq:alpha_hat_closed_form} and~\eqref{eq:st_alpha} using this threshold with the original and transformed calibration scores, respectively, yields the per-input miscoverage estimates $\hat{\alpha}^{\mathrm{BCP}}(x,\mathcal D)$ and $\hat{\alpha}^{\mathrm{ST\text{-}BCP}}(x,\mathcal D)$, for the set $\mathcal C(x)$, both satisfying the reliability guarantee in~\eqref{eq:cvg_guar_e}.

\subsection{$\delta$-PAC-Based Solution}
\label{subsec:pac_bf}
While BCP imposes the same size limit on every input, $\delta$-PAC CP controls the average set size directly through its empirical counterpart in~\eqref{eq:constraint_rule_qd}, allowing the candidate-set size to vary across inputs.

Using the same beam ordering $\{y^{(b)}\}_{b=1}^{B}$ as in Sec.~\ref{subsec:bcp_bf}, we adopt the randomized APS score~\cite{romano2020classification}
\begin{equation}
    s(x,y^{(b)};U) = \sum_{j=1}^{b-1}\hat{p}(y^{(j)}|x) + U\hat{p}(y^{(b)}|x),
    \label{eq:score_bf_pac}
\end{equation}
where $U\sim\operatorname{Unif}[0,1]$ is an auxiliary random variable independent of calibration data $\mathcal D$ and the test pair $(x,y)$. To determine the threshold $q(\mathcal D)$ in~\eqref{eq:constraint_rule_qd}, we evaluate the scores $\{s(x_i,y^{(b)};U_i)\}_{i=1,b=1}^{N_{\mathrm{cal}},B}$ for all calibration inputs and beams. Let $s^{(\ell)}$ denote the $\ell$-th smallest of these scores. Accordingly, the threshold obtained by solving~\eqref{eq:constraint_rule_qd} is
\begin{equation}
    q(\mathcal D)
    =
    s^{\left(\left\lfloor N_{\mathrm{cal}}K\right\rfloor+1\right)}.
    \label{eq:q_bf_pac}
\end{equation}
Evaluating~\eqref{eq:pac_alpha_hat} using the calibration scores $\{s(x_i,y^{(b)};U_i)\}_{i=1}^{N_{\mathrm{cal}}}$ at the threshold $q(\mathcal D)$ in~\eqref{eq:q_bf_pac} yields the $\delta$-PAC miscoverage estimate $\hat{\alpha}^{\mathrm{PAC}}(\mathcal D)$ for the resulting candidate beam sets, with the reliability guarantee in~\eqref{eq:pac_guar}.

\begin{figure*}[!t]
    \centering

    \begin{minipage}[t]{0.5\textwidth}
        \centering
        \includegraphics[width=\linewidth]
        {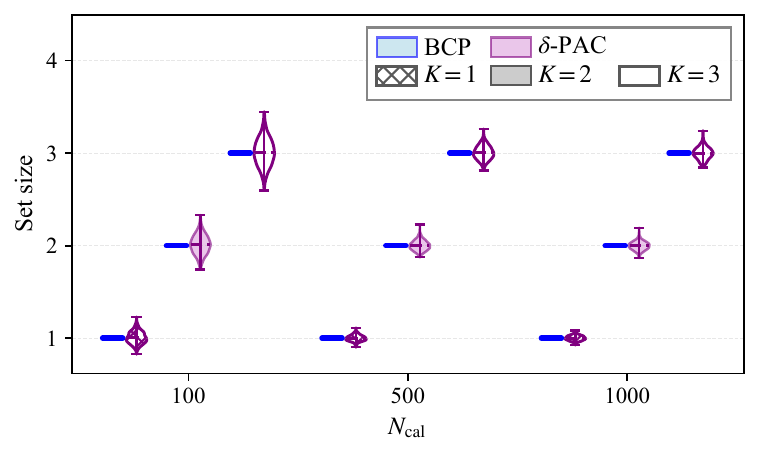}

        \captionof{figure}{Average candidate-set size~\eqref{eq:bf_avg_set_size} versus
        the calibration size $N_{\mathrm{cal}}$ for budgets
        $K\in\{1,2,3\}$.}
        \label{fig:bf_set_size}
    \end{minipage}
    \par
    \begin{minipage}[t]{\textwidth}
        \centering

        \begin{subfigure}[t]{0.49\linewidth}
            \centering
            \includegraphics[width=\linewidth]
            {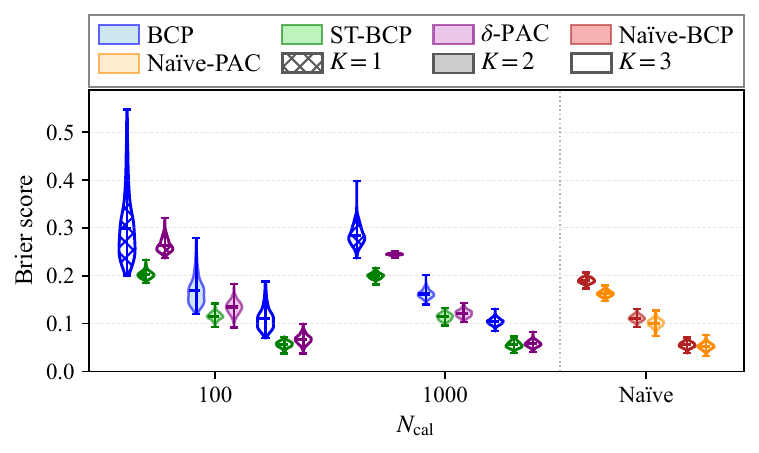}
            \caption{Brier score.}
            \label{fig:bf_brier}
        \end{subfigure}
        \hfill
        \begin{subfigure}[t]{0.49\linewidth}
            \centering
            \includegraphics[width=\linewidth]
            {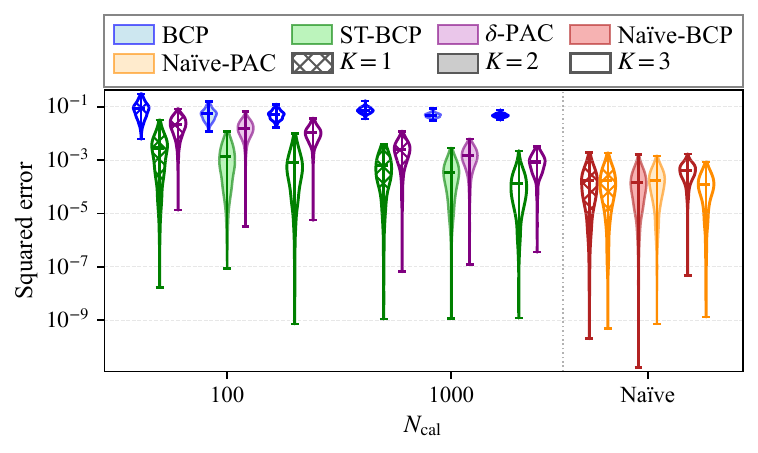}
            \caption{Squared gap.}
            \label{fig:bf_sg}
        \end{subfigure}

        \captionof{figure}{Accuracy of the miscoverage estimates versus the calibration
        size $N_{\mathrm{cal}}$ for budgets $K\in\{1,2,3\}$.}
        \label{fig:bf_accuracy}
    \end{minipage}

\end{figure*}

\subsection{Simulation Results}



\subsubsection{Simulation Setup}
Following the simulation setup in~\cite{othman2026diffusion}, we evaluate codebook-based beam identification using a dataset generated from the ray-traced DeepMIMO ASU outdoor scenario~\cite{alkhateeb2019deepmimo}. The dataset contains $85{,}157$ UE samples, with $90\%$ used for training and the remaining $10\%$ for calibration and testing. The BS employs a discrete Fourier transform codebook with $B=8$ beams. To generate the predictive beam prior $\hat{p}(y|x)$, we adopt the conditional diffusion model in~\cite{othman2026diffusion} with a multilayer perceptron denoiser, modifying only the training procedure by applying early stopping.

We evaluate BCP and ST-BCP in Sec.~\ref{subsec:bcp_bf}, and $\delta$-PAC CP in Sec.~\ref{subsec:pac_bf}, together with two na\"{\i}ve baselines, na\"{\i}ve-BCP and na\"{\i}ve-PAC, obtained by applying the na\"{\i}ve approach in Sec.~\ref{subsec:naive} to the candidate beam sets constructed by BCP and $\delta$-PAC CP, respectively. Unless otherwise specified, we set $\delta=0.1$ for $\delta$-PAC CP, and defer comparisons across different values of $\delta$ to Appendix~\ref{app:pac_delta}.

\subsubsection{Performance Metrics}
For all considered methods, performance is evaluated on the test dataset
$\mathcal{D}^{\mathrm{te}}=\{(x_j,y_j)\}_{j=1}^{N_{\mathrm{te}}}$, with
$N_{\mathrm{te}}=1000$, in terms of reliability, set size, and accuracy.
In each experiment, let $\hat{\alpha}^{\mathrm{PAC}}$ denote the miscoverage estimate for the $\delta$-PAC candidate beam sets, and let $\hat{\alpha}_j^{\mathrm{BCP}}$ and $m_j=\mathbbm{1}(y_j\notin\mathcal C(x_j))$ denote the miscoverage estimate for the BCP set $\mathcal C(x_j)$ and the true miscoverage indicator for the $j$-th test sample, respectively.

Based on these quantities, we assess reliability separately for BCP, ST-BCP and $\delta$-PAC CP according to the guarantees in~\eqref{eq:cvg_guar_e} and~\eqref{eq:pac_guar}, respectively. Specifically, BCP, ST-BCP and na\"{\i}ve-BCP
are evaluated using the reliability ratio in~\eqref{eq:metric_rr},
whereas $\delta$-PAC CP and na\"{\i}ve-PAC are evaluated using the
miscoverage gap
\begin{equation}
    \text{Miscoverage}~\text{gap}
    =
    \hat{\alpha}^{\mathrm{PAC}}
    -
    \frac{1}{N_{\mathrm{te}}}
    \sum_{j=1}^{N_{\mathrm{te}}}m_{j}.
    \label{eq:bf_pac_gap}
\end{equation}
The PAC guarantee requires the gap in~\eqref{eq:bf_pac_gap} to be nonnegative in at least a $1-\delta$ fraction of experiments.

Beyond reliability, we evaluate the probing overhead through the average
candidate-set size
\begin{equation}
    \text{Set size}
    =
    \frac{1}{N_{\mathrm{te}}}
    \sum_{j=1}^{N_{\mathrm{te}}}
    \left|\mathcal C(x_j)\right|.
    \label{eq:bf_avg_set_size}
\end{equation}
For the set-size metric~\eqref{eq:bf_avg_set_size}, we report only BCP and $\delta$-PAC CP, since the corresponding na\"{\i}ve baselines use the same candidate sets.

Finally, we evaluate accuracy using the Brier score in~\eqref{eq:metric_bs}
and the squared gap
\begin{equation}
    \text{Squared gap}
    =
    \left(
    \frac{1}{N_{\mathrm{te}}}
    \sum_{j=1}^{N_{\mathrm{te}}}
    (m_j-\hat{\alpha}_j)
    \right)^2,
    \label{eq:metric_sg}
\end{equation}
where $\hat{\alpha}_j=\hat{\alpha}_j^{\mathrm{BCP}}$ for BCP, while
$\hat{\alpha}_j=\hat{\alpha}^{\mathrm{PAC}}$ for all $j$ under
$\delta$-PAC CP.
We average these metrics over $N_{\mathrm{run}}=500$ independent experiments. In each experiment, disjoint calibration and test sets $\{\mathcal{D}^{\mathrm{cal}},\mathcal{D}^{\mathrm{te}}\}$ are randomly drawn from the held-out dataset.

\subsubsection{Performance Analysis}

Fig.~\ref{fig:bf_reliability} evaluates the reliability of all considered
schemes. Fig.~\ref{fig:bf_reliability_ratio} reports the reliability
ratio~\eqref{eq:metric_rr} for BCP, ST-BCP, and na\"{\i}ve-BCP using the
candidate sets determined by~\eqref{eq:q_bf_bcp}. Across all considered
values of $K$, the average ratios of BCP and ST-BCP remain below one, validating their reliability in the sense
of~\eqref{eq:cvg_guar_e}. In contrast, na\"{\i}ve-BCP frequently underestimates the true miscoverage
and yields reliability ratios substantially greater than one.
Fig.~\ref{fig:bf_pac_gap} presents the miscoverage
gap~\eqref{eq:bf_pac_gap} for $\delta$-PAC CP and na\"{\i}ve-PAC using
the candidate sets determined by~\eqref{eq:q_bf_pac}. For $\delta$-PAC
CP, the lower $\delta$-quantiles remain nonnegative across all values of
$K$, empirically validating the guarantee in~\eqref{eq:pac_guar}. The
corresponding quantiles for na\"{\i}ve-PAC can be negative, indicating
violations of the PAC reliability requirement.

Fig.~\ref{fig:bf_set_size} evaluates whether the candidate sets
constructed by BCP and $\delta$-PAC CP satisfy the average probing
budget $K$. The BCP sets satisfy this budget exactly by construction,
since each set contains $K$ beams. For the $\delta$-PAC sets, the mean
average size also equals $K$, while its variation across experiments
decreases as $N_{\mathrm{cal}}$ increases. This demonstrates increasingly
stable control of the average set size with more calibration data.

Fig.~\ref{fig:bf_accuracy} compares the accuracy of the miscoverage
estimates. In Fig.~\ref{fig:bf_brier}, BCP yields the largest Brier
scores, reflecting its conservative estimates. ST-BCP mitigates this
conservativeness, slightly outperforming $\delta$-PAC CP and approaching
the na\"{\i}ve baselines. This advantage stems from its input-dependent
estimates, whereas $\delta$-PAC CP uses a single estimate for all test
inputs. The squared gap results in Fig.~\ref{fig:bf_sg} indicate that both ST-BCP and
$\delta$-PAC CP achieve low estimation errors, with ST-BCP yielding a
slightly lower squared gap and performing closer to the na\"{\i}ve baselines.
To contextualize ST-BCP's slight accuracy advantage, we further
compare the considered methods across different
$\delta$ values in Appendix~\ref{app:pac_delta}, showing that ST-BCP fails the corresponding PAC guarantee
for most considered $\delta$ values.



\section{Conclusion}
\label{sec:conclusion}

This paper developed a post-hoc conformal framework for reliable uncertainty quantification under operational constraints with applications to wireless systems. Operationally constrained prediction sets are first constructed according to application-specific resource budgets or performance requirements, and their miscoverage is then reliably quantified. For input-dependent constraints, BCP and ST-BCP provide input-specific miscoverage estimates with average-ratio reliability guarantees, while $\delta$-PAC CP provides high-probability reliability guarantees when the constraint is imposed on average. The framework has been evaluated across diverse wireless applications, including narrowband interference detection, near-field localization, and codebook-based beam identification. Numerical results have validated these reliability guarantees and shown that the proposed framework maintains estimation accuracy comparable to that of the na\"{\i}ve approach across the considered applications. Future work may extend the proposed framework to distributed wireless settings involving multiple network nodes with decentralized data.


%

\appendix
\begin{figure*}[!t]
    \centering

    \begin{minipage}[t]{\textwidth}
        \vspace{0pt}
        \centering

        \includegraphics[width=0.5\linewidth]
        {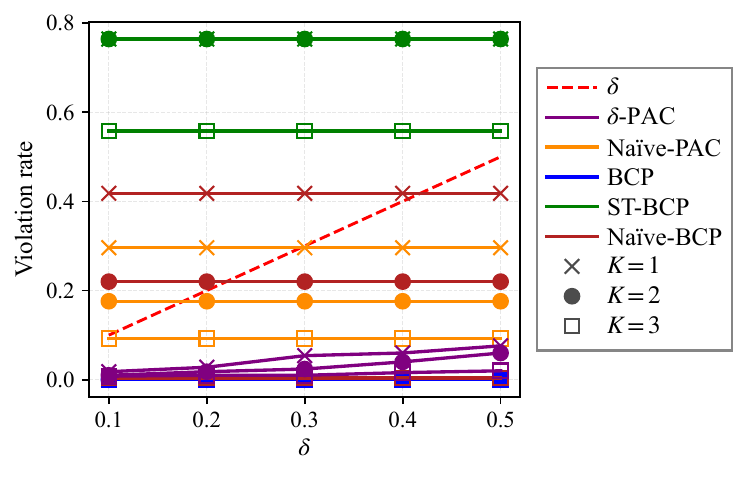}

        \captionof{figure}{Violation rate~\eqref{eq:pac_violation_rate} versus the
        failure level $\delta$ for $N_{\mathrm{cal}}=1000$ and budgets
        $K\in\{1,2,3\}$. The dashed line represents the target violation
        rate $\delta$.}
        \label{fig:pac_delta_violation}
    \end{minipage}
    \par
    \begin{minipage}[t]{\textwidth}
        \vspace{0pt}
        \centering

        \refstepcounter{figure}
        \setcounter{subfigure}{0}

        \begin{subfigure}[t]{0.49\linewidth}
            \vspace{0pt}
            \centering
            \includegraphics[width=\linewidth]
            {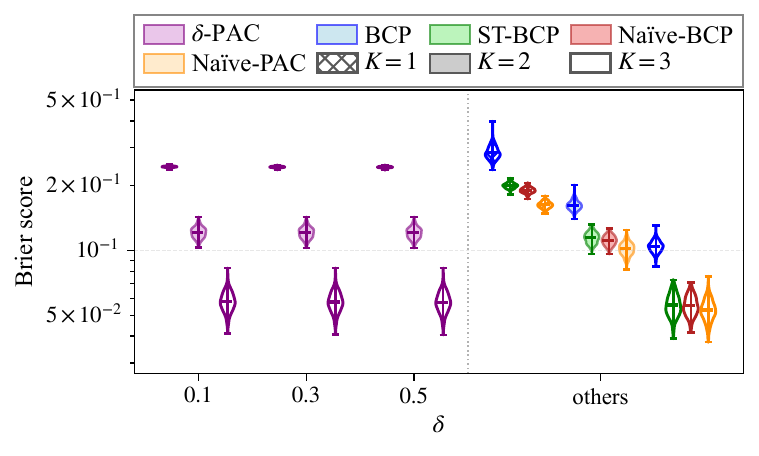}
            \caption{Brier score~\eqref{eq:metric_bs}.}
            \label{fig:pac_delta_brier}
        \end{subfigure}
        \hfill
        \begin{subfigure}[t]{0.49\linewidth}
            \vspace{0pt}
            \centering
            \includegraphics[width=\linewidth]
            {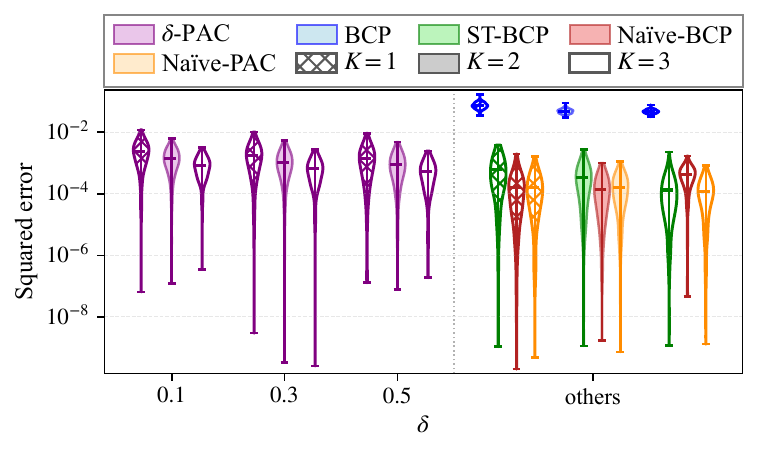}
            \caption{Squared gap~\eqref{eq:metric_sg}.}
            \label{fig:pac_delta_sg}
        \end{subfigure}

        \addtocounter{figure}{-1}
        \captionof{figure}{Accuracy of the miscoverage estimates versus the failure level $\delta$ for $N_{\mathrm{cal}}=1000$ and budgets $K\in\{1,2,3\}$.}
        \label{fig:pac_delta_accuracy}

    \end{minipage}

\end{figure*}

\subsection{Proof of Proposition~\ref{prop:bcp_closed_form_validity}}
\label{app:bcp_proof}

For the test input $x$ and candidate output $y$, define the e-variable~\cite{grunwald2024beyond}
\begin{equation}
    E(x,y)
    =
    \frac{
        s(x,y)
    }{
        \frac{1}{N_{\mathrm{cal}}+1}
        \left(
        \sum_{i=1}^{N_{\mathrm{cal}}} s(x_i,y_i)
        +
        s(x,y)
        \right)
    }.
    \label{eq:e_value}
\end{equation}
Using the score threshold
$q(x)$ in \eqref{eq:constraint_rule_qx}, define the corresponding threshold
e-value as
\begin{equation}
    E^{\mathrm{BCP}}(x,\mathcal D)
    =
    \frac{
        q(x)
    }{
        \frac{1}{N_{\mathrm{cal}}+1}
        \left(
        \sum_{i=1}^{N_{\mathrm{cal}}} s(x_i,y_i)
        +
        q(x)
        \right)
    }.
    \label{eq:bcp_threshold_evalue}
\end{equation}
Since $E(x,y)$ in \eqref{eq:e_value} is monotone increasing in $s(x,y)$, the
score-thresholded representation in \eqref{eq:score_threshold_set} implies
that the prediction set $\mathcal C(x)$ in \eqref{eq:threshold_set_qx} with threshold $q(x)$ can be equivalently
written as
\begin{equation}
    \mathcal C(x)
    =
    \left\{
    y\in\mathcal Y:
    E(x,y)
    <
    E^{\mathrm{BCP}}(x,\mathcal D)
    \right\}.
    \label{eq:app_set_evalue}
\end{equation}

Following~\cite{gauthier2025values,gauthier2025backward},
$\hat{\alpha}^{\mathrm{BCP}}(x,\mathcal D)$ is implicitly defined as
\begin{equation}
\label{eq:alpha_hat_implicit}
\begin{aligned}
\hat{\alpha}^{\mathrm{BCP}}(x,\mathcal D)
=
\inf_{\alpha >0}
&\; \alpha \\
\mathrm{s.t.}
&\;
g\left(
\left\{
y\in\mathcal Y:
E(x,y)<\frac{1}{\alpha}
\right\}
\right)
\le
G_{\max}.
\end{aligned}
\end{equation}
Using the threshold e-value in \eqref{eq:bcp_threshold_evalue}, the infimum
in \eqref{eq:alpha_hat_implicit} is attained when
\begin{equation}
    \hat{\alpha}^{\mathrm{BCP}}(x,\mathcal D)
    =
    \frac{1}{E^{\mathrm{BCP}}(x,\mathcal D)}.
    \label{eq:app_alpha_inverse}
\end{equation}
Substituting \eqref{eq:bcp_threshold_evalue} into
\eqref{eq:app_alpha_inverse} yields the closed-form expression in
\eqref{eq:alpha_hat_closed_form}.

From \eqref{eq:app_set_evalue} and \eqref{eq:app_alpha_inverse}, we have the
equivalence relationship
\begin{equation}
    y
    \notin
    \mathcal C_{q(x)}(x)
    \Longleftrightarrow
    E(x,y)
    \ge
    \frac{1}{
    \hat{\alpha}^{\mathrm{BCP}}(x,\mathcal D)
    }.
    \label{eq:app_miscov_evalue}
\end{equation}
For any positive data-dependent miscoverage level $\tilde{\alpha}>0$ that
may depend on $x$ and $\mathcal D$, the e-variable in \eqref{eq:e_value}
satisfies the post-hoc validity property~\cite{koning2023post}
\begin{equation}
    \mathbb{E}
    \left[
    \frac{
    \mathbbm{1}
    \left(
    E(x,y)
    \ge
    1/\tilde{\alpha}
    \right)
    }{
    \tilde{\alpha}
    }
    \right]
    \le
    1.
    \label{eq:app_post_hoc_e}
\end{equation}
Substituting
$\tilde{\alpha}=\hat{\alpha}^{\mathrm{BCP}}(x,\mathcal D)$ into
\eqref{eq:app_post_hoc_e} and using \eqref{eq:app_miscov_evalue} directly
yields \eqref{eq:post_hoc}.

Finally, since $\hat{\alpha}^{\mathrm{BCP}}(x,\mathcal D)>0$, \eqref{eq:alpha_hat_closed_form} implies that $\hat{\alpha}^{\mathrm{BCP}}(x,\mathcal D)\in(0,1)$ if and only if $q(x)>N_{\mathrm{cal}}^{-1}\sum_{i=1}^{N_{\mathrm{cal}}}s(x_i,y_i)$. This condition is natural, as it requires the threshold to retain outputs that are at least as plausible as a typical true output in the calibration data.

\subsection{Construction of the D\"umbgen--Wellner Confidence Band}
\label{app:dw_band}

Define the adjusted KL
divergence as
\begin{equation}
\mathcal K_{\nu}(a,b)=N_{\mathrm{cal}}K(a,b)-C_{\nu}(a,b)  ,
\label{eq:adjusted_KL}
\end{equation}
where
$K(a,b)=a\log(a/b)+(1-a)\log((1-a)/(1-b))$,
$C_{\nu}(a,b)=\min_{\min\{a,b\}\le t\le\max\{a,b\}}
\{C(t)+\nu D(t)\}$,
$C(t)=\log(\log(e/(4t(1-t))))$, and
$D(t)=\log(1+C(t)^2)$.
Based on \eqref{eq:empirical_cdf} and \eqref{eq:adjusted_KL}, the DW statistic is defined as~\cite{dumbgen2023new}
\begin{equation}
    T_{N_{\mathrm{cal}},\nu}^{\mathrm{DW}}
    =
    \sup_{q\in\mathbb R}
    \left\{
    \mathcal K_{\nu}
    \left(
    \hat F(q;\mathcal D),
    F(q)
    \right)
    \right\}.
    \label{eq:dw_statistic}
\end{equation}
Let $\kappa_{N_{\mathrm{cal}}}^{\mathrm{DW}}(\delta)$ denote the
$(1-\delta)$-quantile of the statistic in \eqref{eq:dw_statistic}. When
$F(\cdot)$ is continuous, this quantile is distribution-free and can be
evaluated through Monte-Carlo simulation using standard uniform random
variables~\cite{dumbgen2023new}. We refer to~\cite{dumbgen2023new} for details on computing
$\kappa_{N_{\mathrm{cal}}}^{\mathrm{DW}}(\delta)$.

Following~\cite{sarkar2023post,dumbgen2023new}, the DW upper confidence bound
is given by
\begin{equation}
\begin{aligned}
u_{\delta}^{\mathrm{DW}}(q;\mathcal D)
&=
\mathcal U_{\delta}^{\mathrm{DW}}
\left(\hat F(q;\mathcal D)\right)
\triangleq
\max\left\{
u\in\left(\hat F(q;\mathcal D),1\right]:
\right.
\\[-1mm]
&\hspace{1.6em}\left.
\mathcal K_{\nu}\left(\hat F(q;\mathcal D),u\right)
\le
\kappa_{N_{\mathrm{cal}}}^{\mathrm{DW}}(\delta)
\right\}.
\end{aligned}
\label{eq:dw_upper_band}
\end{equation}
with $u_{\delta}^{\mathrm{DW}}(q;\mathcal D)=1$ when
$\hat F(q;\mathcal D)=1$.
The corresponding lower confidence bound is obtained by applying the same
upper-bound functional to the complementary empirical CDF level, i.e.,
\begin{equation}
    \ell_{\delta}^{\mathrm{DW}}(q;\mathcal D)
    =
    1
    -
    \mathcal U_{\delta}^{\mathrm{DW}}
    \left(
    1-\hat F(q;\mathcal D)
    \right).
    \label{eq:dw_lower_from_upper}
\end{equation}
Using the symmetry
$\mathcal K_{\nu}(1-a,1-b)=\mathcal K_{\nu}(a,b)$ and substituting
$l=1-u$, \eqref{eq:dw_lower_from_upper} can be equivalently written as
\begin{equation}
    \ell_{\delta}^{\mathrm{DW}}(q;\mathcal D)
    =
    \min
    \left\{
    l\in
    \left[
    0,
    \hat F(q;\mathcal D)
    \right):
    \mathcal K_{\nu}
    \left(
    \hat F(q;\mathcal D),
    l
    \right)
    \le
    \kappa_{N_{\mathrm{cal}}}^{\mathrm{DW}}(\delta)
    \right\},
    \label{eq:dw_lower_band}
\end{equation}
with $\ell_{\delta}^{\mathrm{DW}}(q;\mathcal D)=0$ when
$\hat F(q;\mathcal D)=0$.
Thus, \eqref{eq:dw_lower_band} gives the lower confidence bound used in
\eqref{eq:cdf_band}.

\subsection{PAC Miscoverage Estimation Using Conformal Regret Miscoverage Estimate}
\label{app:creme}

The Conformal REgret Miscoverage Estimate (CREME) method presented in~\cite{zhou2025calibrating} first splits the calibration dataset into disjoint subsets $D_1$ and $ D_2$ with $\mathcal D=\mathcal D_1\cup\mathcal D_2$, using $\mathcal D_1$ to select
the threshold $q(\mathcal D_1)$ according to~\eqref{eq:constraint_rule_qd}
and $\mathcal D_2$ to calibrate the corresponding miscoverage.
Let $N_2=|\mathcal D_2|$. Combining Definition~3.3 and Proposition~3.5 in~\cite{zhou2025calibrating} yields
the estimate
\begin{equation}
\hat{\alpha}^{\mathrm{CREME}}(\mathcal D)
=
\min\left\{
\frac{
N_2\left[
1-\hat{F}\left(q(\mathcal D_1);\mathcal D_2\right)
\right]
+
\sqrt{\frac{N_2}{2}\log\frac{2}{\delta}}
+4
}{
N_2+1
},
1
\right\}.
\label{eq:creme_alpha}
\end{equation}
For comparison, since the DKW bound in~\eqref{eq:dkw_lower_band} holds
simultaneously over all thresholds $q$, it can be evaluated at
$q(\mathcal D_1)$ using $\mathcal D_2$. The corresponding estimate
in~\eqref{eq:pac_alpha_hat} is
\begin{equation}
\hat{\alpha}^{\mathrm{DKW}}(\mathcal D)
=
\min\left\{
1-\hat{F}\left(q(\mathcal D_1);\mathcal D_2\right)
+
\sqrt{\frac{\log(2/\delta)}{2N_2}},
1
\right\}.
\label{eq:dkw_alpha}
\end{equation}

When neither estimate in~\eqref{eq:creme_alpha}
and~\eqref{eq:dkw_alpha} is truncated at one, their difference is
\begin{equation}
\hat{\alpha}^{\mathrm{CREME}}(\mathcal D)
-
\hat{\alpha}^{\mathrm{DKW}}(\mathcal D)
=
\frac{
3+\hat F\left(q(\mathcal D_1);\mathcal D_2\right)
-
\sqrt{\frac{\log(2/\delta)}{2N_2}}
}{
N_2+1
}.
\label{eq:creme_dkw_diff}
\end{equation}
Since we have $\hat F(q(\mathcal D_1);\mathcal D_2)\ge0$, the difference
in~\eqref{eq:creme_dkw_diff} is positive for any number of data points $N_2\ge1$ when
$\delta>2e^{-18}\approx3.05\times10^{-8}$.
As truncation at one preserves this ordering, for practical choices of $\delta$, we have $\hat{\alpha}^{\mathrm{CREME}}(\mathcal D)\ge\hat{\alpha}^{\mathrm{DKW}}(\mathcal D)$. Thus, CREME yields a miscoverage estimate no smaller than that obtained using the DKW bound, and we do not further consider CREME in this study.

\subsection{Impact of the Failure Level $\delta$}
\label{app:pac_delta}

We further examine the impact of the failure level $\delta$ on reliability
and accuracy. To assess the PAC criterion for all considered methods,
we define the violation rate as
\begin{equation}
    \mathrm{Violation~rate}
    =
    \frac{1}{N_{\mathrm{run}}}
    \sum_{r=1}^{N_{\mathrm{run}}}
    \mathbbm{1}
    \left\{
    \frac{1}{N_{\mathrm{te}}}
    \sum_{j=1}^{N_{\mathrm{te}}}m_j^{(r)}
    >
    \frac{1}{N_{\mathrm{te}}}
    \sum_{j=1}^{N_{\mathrm{te}}}\hat{\alpha}_j^{(r)}
    \right\},
    \label{eq:pac_violation_rate}
\end{equation}
where the superscript $r$ identifies the experiment. For $\delta$-PAC
CP, the guarantee in~\eqref{eq:pac_guar} requires the violation rate
in~\eqref{eq:pac_violation_rate} to be no larger than $\delta$.

As expected, Fig.~\ref{fig:pac_delta_violation} shows that the violation rate of $\delta$-PAC CP remains below $\delta$ in all considered settings, consistent with its PAC guarantee~\eqref{eq:pac_guar}. BCP also meets the PAC criterion empirically owing to its conservative miscoverage estimates, although it is not guaranteed to satisfy~\eqref{eq:pac_guar}. In contrast, ST-BCP and the na\"{\i}ve methods fail to satisfy the criterion in all considered settings.

Fig.~\ref{fig:pac_delta_accuracy} evaluates accuracy using the Brier
score~\eqref{eq:metric_bs} and squared gap~\eqref{eq:metric_sg}, both
of which decrease slightly for $\delta$-PAC CP as $\delta$ increases. Moreover, $\delta$-PAC CP yields lower errors than BCP across all considered values of $\delta$, while ST-BCP and the na\"{\i}ve methods generally achieve the lowest errors at the cost of failing to satisfy~\eqref{eq:pac_guar} in most considered settings.






\bibliographystyle{IEEEtran}
\bibliography{./bibtex/bib/IEEEabrv,./bibtex/bib/IEEEreference}

@IEEEtranBSTCTL{IEEEexample:BSTcontrol,
  CTLuse_forced_etal       = "yes",
  CTLmax_names_forced_etal = "3",
  CTLnames_show_etal       = "3",
  CTLdash_repeated_names   = "no"
}

@inproceedings{robinson2023narrowband,
  title={Narrowband interference detection via deep learning},
  author={Robinson, Clifton Paul and Uvaydov, Daniel and D'Oro, Salvatore and Melodia, Tommaso},
  booktitle={ICC 2023-IEEE Int. Conf. Commun.},
  pages={6379--6384},
  year={2023},
  organization={IEEE}
}

@article{huang2025calibrating,
  title={Calibrating Bayesian learning via regularization, confidence minimization, and selective inference},
  author={Huang, Jiayi and Park, Sangwoo and Simeone, Osvaldo},
  journal={IEEE Transactions on Signal Processing},
  volume={73},
  pages={4492--4505},
  year={2025},
  publisher={IEEE}
}

@article{gauthier2025backward,
  title={Backward Conformal Prediction},
  author={Gauthier, Etienne and Bach, Francis and Jordan, Michael I},
  journal={arXiv preprint arXiv:2505.13732},
  year={2025}
}

@article{gauthier2025values,
  title={E-values expand the scope of conformal prediction},
  author={Gauthier, Etienne and Bach, Francis and Jordan, Michael I},
  journal={arXiv preprint arXiv:2503.13050},
  year={2025}
}

@article{angelopoulos2021gentle,
  title={A gentle introduction to conformal prediction and distribution-free uncertainty quantification},
  author={Angelopoulos, Anastasios N and Bates, Stephen},
  journal={Found. Trends Mach. Learn.},
  volume={16},
  number={4},
  pages={494--591},
  year={2023},
  month ={Jul.},
}

@article{cohen2023calibrating,
  title={Calibrating {AI} models for wireless communications via conformal prediction},
  author={Cohen, Kfir M and Park, Sangwoo and Simeone, Osvaldo and Shitz, Shlomo Shamai},
  journal={IEEE Trans. Mach. Learn. Commun. Netw.},
  volume={1},
  pages={296--312},
  year={2023},
  publisher={IEEE}
}

@book{vovk2005algorithmic,
  title={Algorithmic learning in a random world},
  author={Vovk, Vladimir and Gammerman, Alexander and Shafer, Glenn},
  volume={29},
  year={2005},
  month ={Mar.},
  publisher={Springer}
}

@inproceedings{guo2017calibration,
  title={On calibration of modern neural networks},
  author={Guo, Chuan and Pleiss, Geoff and Sun, Yu and Weinberger, Kilian Q},
  booktitle={Proc. Int. Conf. Mach. Learn. (ICML)},
  pages={1321--1330},
  year={2017},
  organization={PMLR}
}

@article{zecchin2023robust,
  author  = {M. Zecchin and S. Park and O. Simeone and M. Kountouris and D. Gesbert},
  title   = {Robust Bayesian learning for reliable wireless {AI}: Framework and applications},
  journal = {IEEE Trans. Cogn. Commun. Netw.},
  volume  = {9},
  number  = {4},
  pages   = {897--912},
  year    = {2023}
}

@article{mozaffarikhosravi2025localization,
  title={Localization-Based Beam Focusing in Near-Field Communications},
  author={Mozaffarikhosravi, Nima and Dharmawansa, Prathapasinghe and Atzeni, Italo},
  journal={IEEE Wireless Commun. Lett.},
  volume={15},
  pages={795--799},
  year={2025},
  publisher={IEEE}
}

@article{sarkar2023post,
  title={Post-selection inference for conformal prediction: Trading off coverage for precision},
  author={Sarkar, Siddhaarth and Kuchibhotla, Arun Kumar},
  journal={arXiv preprint arXiv:2304.06158},
  year={2023}
}

@article{massart1990tight,
  title={The tight constant in the Dvoretzky--Kiefer--Wolfowitz inequality},
  author={Massart, Pascal},
  journal={Ann. Probab.},
  volume={18},
  number={3},
  pages={1269--1283},
  year={1990}
}

@article{dumbgen2023new,
  title={A new approach to tests and confidence bands for distribution functions},
  author={D{\"u}mbgen, Lutz and Wellner, Jon A.},
  journal={Ann. Statist.},
  volume={51},
  number={1},
  pages={260--289},
  year={2023},
  publisher={Institute of Mathematical Statistics}
}

@article{koning2023post,
  title={Post-hoc $\alpha$ Hypothesis Testing and the Post-hoc $ p $-value},
  author={Koning, Nick W},
  journal={arXiv preprint arXiv:2312.08040},
  year={2023}
}

@article{koning2026right,
  title={The`Right'Extension of Type-{I} Error to Data-Dependent Levels},
  author={Koning, Nick W},
  journal={arXiv preprint arXiv:2605.28429},
  year={2026}
}

@article{kiyani2024length,
  title={Length Optimization in Conformal Prediction},
  author={Kiyani, S. and Pappas, G. and Hassani, H.},
  journal={Adv. Neural Inf. Process. Syst.},
  volume={37},
  pages={99519--99563},
  year={2024}
}

@article{koning2025fuzzy,
  title={Fuzzy prediction sets: Conformal prediction with e-values},
  author={Koning, Nick W and van Meer, Sam},
  journal={arXiv preprint arXiv:2509.13130},
  year={2025}
}

@article{romano2020classification,
  title={Classification With Valid and Adaptive Coverage},
  author={Romano, Y. and Sesia, M. and Cand{\`e}s, E.},
  journal={Adv. Neural Inf. Process. Syst.},
  volume={33},
  pages={3581--3591},
  year={2020}
}

@article{grunwald2024beyond,
  title={Beyond {Neyman--Pearson}: E-Values Enable Hypothesis Testing With a Data-Driven Alpha},
  author={Gr{\"u}nwald, P. D.},
  journal={Proc. Nat. Acad. Sci.},
  volume={121},
  number={39},
  pages={e2302098121},
  year={2024}
}

@article{selvan2017fraunhofer,
  title={{Fraunhofer} and {Fresnel}: Unified derivation for aperture antennas},
  author={Selvan, Krishnasamy T and Janaswamy, Ramakrishna},
  journal={IEEE Antennas Propag. Mag.},
  volume={59},
  number={4},
  pages={12--15},
  year={2017},
  publisher={IEEE}
}

@article{gavriilidis2025nearfield,
  author  = {Panagiotis Gavriilidis and George C. Alexandropoulos},
  title   = {Near-Field Beam Tracking With Extremely Large Dynamic Metasurface Antennas},
  journal = {IEEE Trans. Wireless Commun.},
  year    = {2025},
  volume  = {24},
  number  = {7},
  pages   = {6257--6272},
}

@article{su2026reliable,
  title={Reliable Narrowband Interference Detection via Backward Conformal Prediction},
  author={Su, Xin and Zhu, Meiyi and Simeone, Osvaldo and Di Renzo, Marco and Fischione, Carlo},
  journal={arXiv preprint arXiv:2605.02486},
  year={2026}
}

@article{othman2026diffusion,
  title={Diffusion-Based Generative Priors for Efficient Beam Alignment in Directional Networks},
  author={Othman, Esraa Fahmy and Bariah, Lina and Debbah, Merouane},
  journal={arXiv preprint arXiv:2604.09653},
  year={2026}
}

@article{orimogunje2026sensing,
  title={Sensing-Assisted Adaptive Beam Probing With Calibrated Multimodal Priors and Uncertainty-Aware Scheduling},
  author={Orimogunje, Abidemi and Ninkovic, Vukan and Kundacina, Ognjen and Park, Hyunwoo and Kim, Sunwoo and Vukobratovic, Dejan and Twahirwa, Evariste and Gashema, Gaspard},
  journal={IEEE Wireless Commun. Lett.},
  volume={15},
  pages={2438--2442},
  year={2026},
  publisher={IEEE}
}

@inproceedings{liu2026improving,
  title={Improving Backward Conformal Prediction via Non-Conformity Score Transformation},
  author={Liu, Junxian and Zeng, Hao and Wei, Hongxin},
  booktitle={Proc. 43rd Int. Conf. Mach. Learn. (ICML)},
  year={2026}
}

@inproceedings{kiyani2025decision,
  title={Decision Theoretic Foundations for Conformal Prediction: Optimal Uncertainty Quantification for Risk-Averse Agents},
  author={Kiyani, Shayan and Pappas, George J. and Roth, Aaron and Hassani, Hamed},
  booktitle={Proc. 42nd Int. Conf. Mach. Learn. (ICML)},
  pages={30943--30965},
  year={2025}
}

@article{alkhateeb2019deepmimo,
  title={Deep{MIMO}: A generic deep learning dataset for millimeter wave and massive {MIMO} applications},
  author={Alkhateeb, Ahmed},
  journal={arXiv preprint arXiv:1902.06435},
  year={2019}
}

@article{same2021multiple,
  author  = {Same, Mohammad Hossein and Gleeton, Gabriel and
             Gandubert, Gabriel and Ivanov, Preslav and
             Landry, Jr., Ren{\'e}},
  title   = {Multiple Narrowband Interferences Characterization,
             Detection and Mitigation Using Simplified {Welch}
             Algorithm and Notch Filtering},
  journal = {Appl. Sci.},
  volume  = {11},
  number  = {3},
  pages   = {1331},
  year    = {2021}
}

@article{sheemar2025joint,
  title={{Joint holographic beamforming and user scheduling with individual {QoS} constraints}},
  author={Sheemar, Chandan Kumar and Thomas, Christo Kurisummoottil and Alexandropoulos, George C. and Querol, Jorge and Chatzinotas, Symeon and Saad, Walid},
  journal={IEEE Trans. Veh. Tech.},
  volume={74},
  number={10},
  pages={15963--15977},
  year={2025}
}

@inproceedings{bedeer2012adaptive,
  author    = {E. Bedeer and M. Marey and O. A. Dobre and K. E. Baddour},
  title     = {Adaptive bit allocation for {OFDM} cognitive radio systems with imperfect channel estimation},
  booktitle = {Proc. IEEE Radio Wireless Symp. (RWS)},
  pages     = {359--362},
  year      = {2012}
}

@article{simeone2026decision,
  title={Decision Making Needs Uncertainty Quantification [Lecture Notes]},
  author={Simeone, Osvaldo},
  journal={arXiv preprint arXiv:2607.14407},
  year={2026}
}

@article{zhou2025calibrating,
  title={Calibrating decision robustness via inverse conformal risk control},
  author={Zhou, Wenbin and Zhu, Shixiang},
  journal={arXiv preprint arXiv:2510.07750},
  year={2025}
}

@article{garcia2021tutorial,
  title={A Tutorial on 5G NR V2X Communications},
  author={M. H. C. Garcia and A. Molina-Galan and M. Boban and J. Gozalvez and B. Coll-Perales and T. {\c{S}}ahin and A. Kousaridas},
  journal={IEEE Commun. Surveys Tuts.},
  volume={23},
  number={3},
  pages={1972--2026},
  year={2021}
}

@article{ye2018machine,
  title={Machine learning for vehicular networks: Recent advances and application examples},
  author={Ye, Hao and Liang, Le and Li, Geoffrey Ye and Kim, JoonBeom and Lu, Lu and Wu, May},
  journal={IEEE Veh. Technol. Mag.},
  volume={13},
  number={2},
  pages={94--101},
  year={2018}
}

@inproceedings{masegosa2020learning,
  author    = {A. Masegosa},
  title     = {Learning under Model Misspecification: Applications to Variational and Ensemble Methods},
  booktitle = {Adv. Neural Inf. Process. Syst.},
  volume    = {33},
  pages     = {5479--5491},
  year      = {2020}
}

@article{popovski2018wireless,
  title={Wireless Access for Ultra-Reliable Low-Latency Communication: Principles and Building Blocks},
  author={Popovski, Petar and Nielsen, Jimmy J and Stefanovic, Cedomir and De Carvalho, Elisabeth and Strom, Erik and Trillingsgaard, Kasper F and Bana, Alexandru-Sabin and Kim, Dong Min and Kotaba, Radoslaw and Park, Jihong and others},
  journal={IEEE Netw.},
  volume={32},
  number={2},
  pages={16--23},
  year={2018},
  publisher={IEEE}
}

@article{sze2017efficient,
  title={Efficient processing of deep neural networks: A tutorial and survey},
  author={Sze, Vivienne and Chen, Yu-Hsin and Yang, Tien-Ju and Emer, Joel S},
  journal={Proc. IEEE},
  volume={105},
  number={12},
  pages={2295--2329},
  year={2017},
  publisher={IEEE}
}

@inproceedings{su2025conformal,
  title={Conformal Robust Beamforming via Generative Channel Models},
  author={Su, Xin and Hou, Qiushuo and He, Ruisi and Simeone, Osvaldo},
  booktitle={Proc. IEEE Int. Workshop Signal Process. Artif. Intell. Wireless Commun. (SPAWC)},
  pages={1--5},
  year={2025},
}

@article{chenreddy2022data,
  title={Data-driven conditional robust optimization},
  author={Chenreddy, Abhilash Reddy and Bandi, Nymisha and Delage, Erick},
  journal={Adv. Neural Inf. Process. Syst.},
  volume={35},
  pages={9525--9537},
  year={2022}
}

@article{lu2020positioning,
  title={Positioning-aided {3D} beamforming for enhanced communications in mmWave mobile networks},
  author={Lu, Yi and Koivisto, Mike and Talvitie, Jukka and Valkama, Mikko and Lohan, Elena Simona},
  journal={IEEE Access},
  volume={8},
  pages={55513--55525},
  year={2020},
  publisher={IEEE}
}

@book{simeone2022machine,
  author    = {O. Simeone},
  title     = {Machine Learning for Engineers},
  publisher = {Cambridge Univ. Press},
  address   = {Cambridge, U.K.},
  year      = {2022}
}

@article{tedeschini2024real,
  author  = {B. C. Tedeschini and G. Kwon and M. Nicoli and M. Z. Win},
  title   = {Real-time {Bayesian} neural networks for {6G} cooperative positioning and tracking},
  journal = {IEEE J. Sel. Areas Commun.},
  vol.    = {42},
  no.     = {9},
  pages   = {2322--2338},
  year    = {2024}
}

@article{zhang2026efficient,
  author  = {Z. Zhang and C. Miao and J. Yang and Y. Chen},
  title   = {Efficient uncertainty quantification for power allocation in wireless networks},
  journal = {J. King Saud Univ. Comput. Inf. Sci.},
  vol.    = {38},
  no.     = {3},
  pages   = {106},
  year    = {2026}
}

@article{dong2020survey,
  author  = {X. Dong and Z. Yu and W. Cao and Y. Shi and Q. Ma},
  title   = {A survey on ensemble learning},
  journal = {Front. Comput. Sci.},
  vol.    = {14},
  no.     = {2},
  pages   = {241--258},
  year    = {2020}
}

@article{jankov2026reliability,
  author  = {M. Jankov and C. Fischione},
  title   = {Reliability-aware scheduling for {Digital Twin} maintenance},
  journal = {arXiv preprint arXiv:2608.21866},
  year    = {2026}
}

@article{sabanovic2026calibration,
  author  = {D. {\v{S}}abanovi{\'c} and T. Kr{\v{c}}mar and Z. Krpi{\'c} and I. Luki{\'c}},
  title   = {Calibration, architecture, and distribution shift in predictive uncertainty estimation},
  journal = {Mach. Learn. Knowl. Extr.},
  vol.    = {8},
  no.     = {7},
  pages   = {179},
  year    = {2026}
}

@incollection{platt1999probabilistic,
  author    = {John C. Platt},
  title     = {Probabilistic Outputs for Support Vector Machines and Comparisons to Regularized Likelihood Methods},
  booktitle = {Advances in Large Margin Classifiers},
  editor    = {Alexander J. Smola and Peter L. Bartlett and Bernhard Sch{\"o}lkopf and Dale Schuurmans},
  pages     = {61--74},
  publisher = {MIT Press},
  address   = {Cambridge, MA, USA},
  year      = {1999}
}

@inproceedings{zadrozny2002transforming,
  author    = {Bianca Zadrozny and Charles Elkan},
  title     = {Transforming Classifier Scores into Accurate Multiclass Probability Estimates},
  booktitle = {Proc. 8th ACM SIGKDD Int. Conf. Knowledge Discovery and Data Mining (KDD)},
  pages     = {694--699},
  address   = {Edmonton, AB, Canada},
  publisher = {ACM},
  year      = {2002},
}

@inproceedings{kull2019beyond,
  author    = {Meelis Kull and Miquel Perello Nieto and Markus K{\"a}ngsepp and Telmo Silva Filho and Hao Song and Peter Flach},
  title     = {Beyond Temperature Scaling: Obtaining Well-Calibrated Multi-Class Probabilities with {Dirichlet} Calibration},
  booktitle = {Adv. Neural Inf. Process. Syst. (NeurIPS)},
  volume    = {32},
  pages     = {12316--12326},
  year      = {2019}
}

@inproceedings{tomani2022parameterized,
  title={Parameterized temperature scaling for boosting the expressive power in post-hoc uncertainty calibration},
  author={Tomani, Christian and Cremers, Daniel and Buettner, Florian},
  booktitle={Proc. Eur. Conf. Comput. Vis. (ECCV)},
  pages={555--569},
  year={2022},
  organization={Springer}
}

@inproceedings{judah2024demonstrating,
  title={Demonstrating confidence in radio frequency machine learning systems},
  author={Judah, Matthew and Kuzdeba, Scott},
  booktitle={Proc. {IEEE} Workshop Signal Process. Syst. ({SiPS})},
  pages={13--18},
  year={2024},
  organization={IEEE}
}

@article{raina2025trust,
  title={To trust or not to trust: On calibration in {ML}-based resource allocation for wireless networks},
  author={Raina, Rashika and Simmons, Nidhi and Simmons, David E and Yacoub, Michel Daoud and Duong, Trung Q},
  journal={{IEEE} Trans. Netw. Sci. Eng.},
  year={2025},
  publisher={IEEE}
}

@inproceedings{cohen2023calibratingc,
  title={Calibrating {AI} Models for Few-Shot Demodulation via Conformal Prediction},
  author={Cohen, Kfir M. and Park, Sangwoo and Simeone, Osvaldo and Shitz, Shlomo Shamai},
  booktitle={Proc. IEEE Int. Conf. Acoust., Speech Signal Process. (ICASSP)},
  pages={1--5},
  year={2023},
}

@article{zhang2025calibrating,
  author  = {J. Zhang and Q. Hou and X. Su and G. Yu},
  title   = {Calibrating {AI}-based beamforming via conformal prediction},
  journal = {IEEE Commun. Lett.},
  vol.    = {30},
  pages   = {56--60},
  year    = {2025}
}

@article{hou2025what,
  title={What if We Had Used a Different App? Reliable Counterfactual {KPI} Analysis in Wireless Systems},
  author={Hou, Qiushuo and Park, Sangwoo and Zecchin, Matteo and Cai, Yunlong and Yu, Guanding and Simeone, Osvaldo},
  journal={IEEE Trans. Cogn. Commun. Netw.},
  volume={11},
  number={5},
  pages={3529--3543},
  year={2025},
  publisher={IEEE}
}

@article{zhu2024federated,
  title={Federated Inference with Reliable Uncertainty Quantification over Wireless Channels via Conformal Prediction},
  author={Zhu, Meiyi and Zecchin, Matteo and Park, Sangwoo and Guo, Caili and Feng, Chunyan and Simeone, Osvaldo},
  journal={IEEE Trans. Signal Process.},
  volume={72},
  pages={1235--1250},
  year={2024},
  publisher={IEEE}
}

@article{zhu2025conformal,
  author  = {M. Zhu and M. Zecchin and S. Park and C. Guo and C. Feng and P. Popovski and O. Simeone},
  title   = {Conformal distributed remote inference in sensor networks under reliability and communication constraints},
  journal = {IEEE Trans. Signal Process.},
  vol.    = {73},
  pages   = {1485--1500},
  year    = {2025}
}



%


\end{document}